\documentclass[10pt, letter, a4paper]{./lib/IEEEtran}
\usepackage{amsmath}
\usepackage{mathtools}
\usepackage{array}
\usepackage{cite}
\usepackage{algpseudocode,algorithm,algpseudocode}
\usepackage{float}
\usepackage{tabularx}
\usepackage{mathrsfs} 
\usepackage{tikz}
\usetikzlibrary{chains, arrows,shapes,positioning,arrows,decorations.markings,calc}
\usepackage{pgfplots}
\pgfplotsset{compat=newest}
\usepgfplotslibrary{groupplots}
\usepgfplotslibrary{fillbetween}
\usepackage{gensymb}
\usepackage{soul}	
\usepackage{xcolor}
\usepackage[shortcuts,acronym]{glossaries}
\makeglossaries
\usepackage{fixmath}
\usepackage{amssymb}
\usepackage{bm}
\usepackage{bbm}
\usepackage{pbox}
\usepackage{subfig}
\usepackage{multirow}
\usepackage{tabularx}
\pgfplotsset{compat=newest,
	/pgfplots/ybar legend/.style={
		/pgfplots/legend image code/.code={%
			\draw[##1,/tikz/.cd,bar width=3pt,yshift=-0.2em,bar shift=0pt]
			plot coordinates {(0cm,0.8em)};},
	},}

\usepackage{amsthm}
\theoremstyle{plain}
\newtheorem{theorem}{Theorem}
\newtheorem{lemma}{Lemma}
\newtheorem{approximation}{Approximation}
\newtheorem{remark}{Remark}

\usepackage{caption}
\makeatletter
\def\therule{\makebox[\algorithmicindent][l]{\hspace*{.5em}\vrule height .75\baselineskip depth .25\baselineskip}}%

\newtoks\therules
\therules={}
\def\appendto#1#2{\expandafter#1\expandafter{\the#1#2}}
\def\gobblefirst#1{
	#1\expandafter\expandafter\expandafter{\expandafter\@gobble\the#1}}%
\def\LState{\State\unskip\the\therules}
\def\pushindent{\appendto\therules\therule}
\def\popindent{\gobblefirst\therules}
\def\printindent{\unskip\the\therules}
\def\printandpush{\printindent\pushindent}
\def\popandprint{\popindent\printindent}

\algdef{SE}[WHILE]{While}{EndWhile}[1]
{\printandpush\algorithmicwhile\ #1\ \algorithmicdo}
{\popandprint\algorithmicend\ \algorithmicwhile}%
\algdef{SE}[FOR]{For}{EndFor}[1]
{\printandpush\algorithmicfor\ #1\ \algorithmicdo}
{\popandprint\algorithmicend\ \algorithmicfor}%
\algdef{S}[FOR]{ForAll}[1]
{\printindent\algorithmicforall\ #1\ \algorithmicdo}%
\algdef{SE}[LOOP]{Loop}{EndLoop}
{\printandpush\algorithmicloop}
{\popandprint\algorithmicend\ \algorithmicloop}%
\algdef{SE}[REPEAT]{Repeat}{Until}
{\printandpush\algorithmicrepeat}[1]
{\popandprint\algorithmicuntil\ #1}%
\algdef{SE}[IF]{If}{EndIf}[1]
{\printandpush\algorithmicif\ #1\ \algorithmicthen}
{\popandprint\algorithmicend\ \algorithmicif}%
\algdef{C}[IF]{IF}{ElsIf}[1]
{\popandprint\pushindent\algorithmicelse\ \algorithmicif\ #1\ \algorithmicthen}%
\algdef{Ce}[ELSE]{IF}{Else}{EndIf}
{\popandprint\pushindent\algorithmicelse}%
\algdef{SE}[PROCEDURE]{Procedure}{EndProcedure}[2]
{\printandpush\algorithmicprocedure\ \textproc{#1}\ifthenelse{\equal{#2}{}}{}{(#2)}}%
{\popandprint\algorithmicend\ \algorithmicprocedure}%
\algdef{SE}[FUNCTION]{Function}{EndFunction}[2]
{\printandpush\algorithmicfunction\ \textproc{#1}\ifthenelse{\equal{#2}{}}{}{(#2)}}%
{\popandprint\algorithmicend\ \algorithmicfunction}%

\makeatother

\algrenewcommand\algorithmicprocedure{\textbf{Input}}
\algrenewcommand\algorithmicreturn{\textbf{Output:}}

\begin{document}
\bstctlcite{IEEEexample:BSTcontrol}

\title{A Spatiotemporal Model for Peak AoI in Uplink IoT Networks: Time Vs Event-triggered Traffic}
	
	\author{Mustafa~Emara,~\IEEEmembership{Student Member,~IEEE,}
	Hesham~ElSawy,~\IEEEmembership{Senior Member,~IEEE,}
	Gerhard~Bauch,~\IEEEmembership{Fellow,~IEEE}
	\thanks{M. Emara is with the Germany standards R\&D team, Next Generation and Standards, Intel Deutschland GmbH and  the Institute of Communications, Hamburg University of Technology, Hamburg, 21073 Germany (e-mail: mustafa.emara@intel.com)}
	\thanks{H. ElSawy is with the Electrical Engineering Department, King Fahd University of Petroleum and Minerals, 31261 Dhahran, Saudi Arabia (email: hesham.elsawy@kfupm.edu.sa).}
	\thanks{G. Bauch is with the Institute of Communications,
		Hamburg University of Technology, Hamburg, 21073 Germany (email: bauch@tuhh.de).}
	\thanks{This paper has been submitted in part to ICC 2020 \cite{Emara2020}.}}
	\maketitle
	\thispagestyle{empty}
	\maketitle
	\thispagestyle{empty}
	
\newacronym{AoI}{AoI}{age of information}
\newacronym{BS}{BS}{base station}
\newacronym{DTMC}{DTMC}{discrete time Markov chain}
\newacronym{ET}{ET}{event-triggered}
\newacronym{5G}{5G}{fifth generation}
\newacronym{FCFS}{FCFS}{first come first serve}
\newacronym{IoT}{IoT}{Internet of Things}
\newacronym{KPI}{KPI}{key performance indicator}
\newacronym{LT}{LT}{Laplace transform}
\newacronym{LTE}{LTE}{long term evolution}
\newacronym{LCFS}{LCFS}{last come first serve}
\newacronym{MAM}{MAM}{matrix analytic method}
\newacronym{MTC}{MTC}{machine type communication}
\newacronym{MAC}{MAC}{medium access control}

\newacronym{NB-IoT}{NB-IoT}{narrowband IoT}
\newacronym{PPP}{PPP}{Poisson point processes}
\newacronym{PP}{PP}{point processes}
\newacronym{PDF}{PDF}{probability density function}
\newacronym{PAoI}{PAoI}{Peak AoI}
\newacronym{QoS}{QoS}{quality of service}
\newacronym{QCI}{QCI}{QoS class identifier}
\newacronym{QBD}{QBD}{quasi-birth-death}
\newacronym{RAT}{RAT}{radio access technology}
\newacronym{SINR}{SINR}{signal to interference noise ratio}
\newacronym{SIR}{SIR}{signal to interference ratio}
\newacronym{3GPP}{3GPP}{Third generation partnership project}
\newacronym{TSN}{TSN}{time senstive networking}
\newacronym{TSP}{TSP}{transmission success probability}
\newacronym{TT}{TT}{time-triggered}
\newacronym{URLLC}{URLLC}{ultra reliable low latency communication}
\newacronym{WA}{WA}{weighted allocation}
	\thispagestyle{empty}
\begin{abstract}
Timely message delivery is a key enabler for Internet of Things (IoT) and cyber-physical systems to support wide range of context-dependent applications. Conventional time-related metrics (e.g. delay and jitter) fails to characterize the timeliness of the system update. Age of information (AoI) is a time-evolving metric that accounts for the packet inter-arrival and waiting times to assess the freshness of information. In the foreseen large-scale IoT networks, mutual interference imposes a delicate relation between traffic generation patterns and transmission delays. To this end, we provide a spatiotemporal framework that captures the peak AoI (PAoI) for large scale IoT uplink network under time-triggered (TT) and event triggered (ET) traffic. Tools from stochastic geometry and queueing theory are utilized to account for the macroscopic and microscopic network scales. Simulations are conducted to validate the proposed mathematical framework and assess the effect of traffic load on PAoI. The results unveil a counter-intuitive superiority of the ET traffic over the TT in terms of PAoI, which is due to the involved temporal interference correlations. Insights regarding the network stability frontiers and the location-dependent performance are presented. Key design recommendations regarding the traffic load and decoding thresholds are highlighted.
\end{abstract}
%

\section{Introduction}\label{section:introduction}

The exploding growth in network traffic is creating data transfer, monitoring, timing, and scaling challenges \cite{3GPP2018}. The timeliness and retainability of continuous updates of nodes within a system are overarching requirements among different technology segments, such as vehicular, industrial \ac{IoT}, and cellular \cite{Kim2012}. This implies continuous information update about the real-time states between a given source and its targeted destination \cite{NGMNA2016}. Focusing on \ac{IoT} and its underlying architecture, that include among others, advanced software components, ubiquitous sensors, autonomous actuators, and a communications network infrastructure \cite{Fuqaha2015, Soltanmohammadi2016}. This allows the devices to communicate with proximate devices and learn from their surrounding environment. One key characterization of \ac{IoT}  is the traffic generated by the \ac{IoT} devices, which governs many of the system key performance indicators. Therefore, it is important to provide a mathematical framework that can characterize the information freshness within large scale uplink \ac{IoT} networks under different  traffic models. 

\subsection{Background}

\ac{IoT} traffic can be categorized into \ac{TT} and \ac{ET} traffic, with typically sporadic and short payload packets \cite{Metzger2019}. \ac{TT} events generate periodic traffic as in vehicular communications, smart grids and wireless sensor networks \cite{Palattella2016}. As an example, one may consider a smart monitoring application where devices send timely-based updates to the network server, comprising the uncoordinated time-triggered (i.e., periodic/deterministic) traffic. In such segments, a central entity collects status updates from multiple nodes (e.g., sensors, vehicles and monitors) through wireless channels. On the other hand, \ac{ET} traffic arises in scenarios where multiple devices transmit based on detected events \cite{Gupta2017}. Such scenarios can be observed as an example in a given area where power outage occurs. Thousands of devices report their status before the outage occurrence. Support of \ac{IoT} network for the two considered traffic models is crucial to maintain network functionality and attain the required \acp{QoS} \cite{3GPP2018vert}. Throughout this work, we address the critical challenge of how to maintain timely status updates over all the connected nodes within an \ac{IoT} uplink network under the two variants of traffic models, namely \ac{TT} and \ac{ET} traffic. 

To characterize the freshness of information at the receiver, we adopt the proposed \ac{AoI} metric in \cite{Kaul2012}, which has received increasing attention in the past years. The age of information accumulates the transmission delay in addition to time elapsed between successive system updates  \cite{Yates2019}. Hence, the \ac{AoI} increases even when there are no packets in the system to account for the freshness of information. Compared to traditional time metrics (e.g. delay and jitter), \ac{AoI} captures the timeliness of updates in a way those traditional metrics do not \cite{Talak2018, Ceran2018}. 

To position our contribution in context, we first discuss a series of key prior works that studied the \ac{AoI} and its variants. Authors in \cite{Kaul2012} consider the system where a sensor generates and transmits update packets to its destination under a \ac{FCFS} principle and derive the expression of average \ac{AoI} for different queueing models. The work in \cite{Kaul2012} is extended to out-of-order packet delivery in \cite{Kam2013}. Last come first serve (LCFS) queue discipline, with and without service preemption, is studied and contrasted to \ac{FCFS} in \cite{Yates2019, Kaul2012b}. The \ac{AoI} is also characterized in \cite{Kaul2018} for prioritized packet delivery and in \cite{Talak2018} for deterministic traffic models. In summary, the aforementioned works consider a single sensor scenario. 

In addition, a number of works have considered the information freshness in \ac{IoT} networks with multiple sensors \cite{Pappas2015, Huang2015, Kaul2018}. In particular, authors in \cite{Pappas2015} consider that one transmitter sends status update packets generated from multiple sensors to the destination, and analyze the average \ac{AoI} for updates allowing the latest arrival to overwrite the previous queued ones. The authors in \cite{Costa2016,He2016} propose a new metric, namely \ac{PAoI}, that characterizes the maximum value of the age achieved immediately before receiving a new packet. Focusing on the \ac{PAoI}, \cite{Huang2015} analyzed the system performance by considering a general service time distribution, and optimized the update arrival rates to minimize its defined \ac{PAoI}-related system cost. In \cite{AoIDhillon}, the authors investigated the role of an unmanned aerial vehicle as a mobile relay to minimize the \ac{PAoI}. The joint effects of data preprocessing and transmission procedures on the \ac{PAoI} under Poisson traffic model was investigated in \cite{Xu2019} without considerations of the network-wide interference. 

While the aforementioned works characterize the \ac{AoI} at the microscopic device level, they overlook the macroscopic impact of aggregate network interference between multiple devices. In the foreseen massively loaded \ac{IoT} networks, the mutual interference between the active transmitters, trying to utilize the set of finite resources, might hinder timely updates of a given link of interest \cite{Ayoub2018}. In the context of large-scale networks, stochastic geometry is a mathematical tool that is is employed to characterize performance when accounting for mutual interference within the network \cite{Andrews2011, Elsawy_tutorial, Haenggi2012}. However, the commonly adopted full-buffer assumption hinders the evaluation of temporal based metrics such as delay and \ac{AoI}. To overcome such limitation, spatiotemporal models are developed. In particular, stochastic geometry and queueing theory are jointly utilized to characterize both the macroscopic network-wide interference and microscopic per-device queue evolution \cite{Zhong2017, Gharbieh2018, Yang2019, Chisci2019, YangAoI2019}. Capitalizing on such spatiotemporal models, delay and \ac{AoI} can be characterized and assessed in large scale massive networks. For instance, lower and upper bounds for the average \ac{AoI} are proposed under a stochastic geometry framework in \cite{Hu2018}. Additionally, \ac{AoI} under a spatiotemporal framework has recently been investigated in \cite{YangAoI2019}, where the authors investigated different scheduling techniques to optimize the \ac{PAoI} under a spatiotemporal framework. 

In summary, the \ac{AoI} in large-scale IoT networks with \ac{TT} traffic has not been studied yet. Furthermore, to the best of the authors knowledge, \ac{AoI} for uplink traffic in large-scale IoT networks is a still an open research problem. To this end, characterizing the \ac{AoI} leads to informed insights on how to enhance the performance of time-critical applications.  
 
\subsection{Contributions}
Throughout this work,\footnote{A simplified version of this work is presented in part in \cite{Emara2020}, which is limited to \ac{ET} traffic only.} we provide a mathematical framework to characterize the information freshness via the \ac{PAoI}. When compared to the average \ac{AoI}, \ac{PAoI} is considered throughout this work because it is more suited to provision \ac{QoS} and for min-max network design objectives \cite{Huang2015, YangAoI2019}. The macroscopic and microscopic scales of an uplink large scale \ac{IoT} network are  addressed through the proposed framework. For the macroscopic aspect, stochastic geometry is utilized to characterize for the mutual interference among active devices (i.e., position dependent). In addition, queueing models from queueing theory are adopted to account for the microscopic queue evolution at each device under the \ac{TT} and \ac{ET} traffic models. To track the queue status at a given time stamp, a \ac{DTMC} is utilized for each device. Expressions for the distribution of the \ac{TSP} are derived, which entails the effect of the considered traffic model. In addition, the \ac{TT} traffic is modeled via an absorbing Markov chain that mimics the duty cycle of the generated packets. In summary, the main contributions of this paper compared to the previously stated works are summarized as follows:
\begin{itemize}
	\item Develop a novel and tractable mathematical framework  that characterizes the spatiotemporal interactions under \ac{TT} and \ac{ET} traffic models; 
	\item Employ a DTMC at every IoT device to track the temporal dynamics for the \ac{TT} and \ac{ET} traffic models;
	\item Integrate the developed \acp{DTMC} with stochastic geometry framework to evaluate the \ac{PAoI} in large-scale IoT networks;
	\item Assess the \ac{PAoI} for the \ac{TT} and \ac{ET} traffic models; and
	\item Showcase the Pareto frontiers that characterize the network's stability regions.
\end{itemize}

\subsection{Notation and Organization}
In this work, the following notation will be adopted. Upper-case and lower-case boldface letters ($\mathbf{A}$, $\mathbf{a}$) represent matrices and vectors, respectively. $\mathbf{1}_m$ and $\mathbf{\mathcal{I}}_m$ denote, respectively, an all ones vector and matrix of dimension $m \times m$. An identity matrix of dimension $m$ is represented via $\mathbf{I}_{m}$. Over the bar operation depicts the complement operator (i.e., $\bar{v} = 1-v$). Furthermore, $\mathbbm{1}_{\{z\}}$ represents the indicator function which equals 1 if the expression $z$ is true and 0 otherwise. The probability of an event and its expectation are given by $\mathbb{P}\{\cdot\}$ and $\mathbb{E}\{\cdot\}$, respectively. 

The rest of the paper is organized as follows. Section \ref{sec:system_model} presents the system model, the underlying physical and \ac{MAC} parameters, and the \ac{PAoI} evaluation. Section \ref{sec:sg_analysis} discusses the location-dependent characterization of the network-wide interference for the \ac{TT} and \ac{ET} traffic models. The \ac{TT} and \ac{ET}  queueing models along with the microscopic intra-device interactions are provided in Section \ref{sec:QT_anaylsis}. In Section \ref{sec:simulation_results}, various simulation results and observations are discussed. Finally, Section \ref{sec:Conclusion} summarizes the work and draw final conclusions.
	\thispagestyle{empty}
\section{System Model}\label{sec:system_model}
\subsection{Spatial \& Physical Layer Parameters}
An uplink cellular network is considered in this work, where the \acp{BS} are deployed based on a \ac{PPP} $\mathrm{\Psi}$ with spatial intensity $\lambda$ BS/km$^2$. The \ac{IoT} devices follow an independent \ac{PPP} $\mathrm{\Phi}$, such that  within the Voronoi cell of every BS $b_i\in\mathrm{\Phi}$, a device is dropped uniformly and independently. All devices and \acp{BS} are equipped with single antennas. Let $r$ be the distance between a device and its serving \ac{BS} and $\eta > 2$ be the path-loss exponent, an unbounded path-loss propagation model is considered such that the signal power attenuates at the rate $r^{-\eta}$. Multi-path Rayleigh fading is assumed to characterize the small-scale fading. Additionally, $h$ and $g$ denote the intended and interference channel power gains, and are exponentially distributed with unit power gain. Spatial and temporal independence is assumed for all the channel gains. Fractional path-loss inversion power control is considered at the devices with compensation factor $\epsilon$. Accordingly, the transmit power of a device positioned $r$ meters is given by $\rho r^{\eta\epsilon}$, where $\rho$ is a power control parameter to adjust the average received power at the serving BS \cite{ElSawy2014}. In this work, a fixed, yet arbitrary network realization of the network is considered to account for the much smaller time scale of the channel fading, packet generation, and transmission when compared to the spatial network dynamics.\footnote{To analyze the location-dependent performance of the network, we consider a static network where for a generic network realization, $\mathrm{\Phi}$ and $\mathrm{\Psi}$ remain static over sufficiently large time horizon, while device activities, channel fading, and queue states vary each time slot.}

\subsection{Temporal \& MAC layer parameters}
The proposed framework studies a discretized, time slotted, and synchronized system, in which a new packet is generated at a generic device based on \ac{TT} or \ac{ET} traffic. For the \ac{TT} traffic, we consider an asynchronous homogeneous periodic packet generation with duty cycle $T$ and time-slot offset $\beta$. That is, each device in the network generates a packet (e.g., measurement or status update) periodically every $T$ time slots. However, it is not necessary that all devices in the network are synchronized to the same time slot for packet generation. Instead, it is assumed that the offset of the devices $\beta_i\in\{0,1,\cdots,T-1\}$, $\forall i \in \mathrm{\Phi}$ are  independently and uniformly distributed among the time slots within the duty cycle $T$, i.e., $\mathbb{P}\{\beta_i = \tau\}=1/T,\; \tau\in\{0,1,\cdots,T-1\}$.  For the \ac{ET} traffic, a Bernoulli traffic model is adopted, in which a new packet is generated at a generic device with an independent slot-wise arrival probability of $\alpha \in (0,1]$. 

A \ac{FCFS} discipline is considered at each device, where failed packets are persistently retransmitted till successful reception. In particular, a packet residing at a generic device is successfully decoded if the received \ac{SIR} is larger than a detection threshold $\theta$ at its serving \ac{BS}. In the case of successful decoding, an ACK is transmitted from the \ac{BS} via an error-free feedback channel so the device can drop this head of the queue packet. In the case of failed decoding, the serving \ac{BS} transmits an NACK and the packet remains at the head of the device's queue and a new transmissions is attempted in the next time slot. Error-free and negligible delay for ACK and NACK is adopted throughout this work. 
\begin{figure}
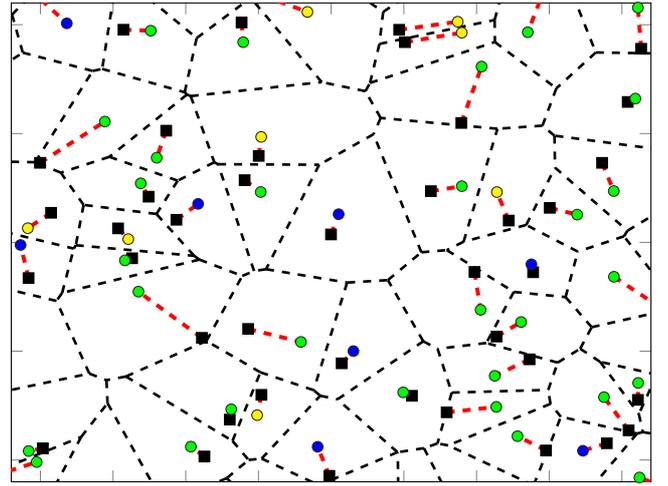

	\begin{center}
		\ifCLASSOPTIONdraftcls
		\input{system_model/figures/deployment/1col/deployment.tex}	
		\else
		\input{system_model/figures/deployment/2col/deployment.tex}	
		\fi			
		\caption{A network realization for $\lambda=10^{-6}$ $\text{BS/KM}^2$, $\theta = 1$ and $T=8$. Black squares depict the BSs while green, blue and yellow circles represent devices with empty queue, devices with the same $\beta$ and devices with residual packets attempting a retransmission due to past failed attempts, respectively. The coverage region of the \acp{BS} and their connected devices are represented by dashed black and red lines, respectively.}
		\label{fig:spatiotemporal_realization}
	\end{center}
\end{figure}

 In Fig. \ref{fig:spatiotemporal_realization}, a spatiotemporal realization of the network is shown. At a given time slot, two different states of devices can be observed i) active due to non-empty queue and ii) idle due to empty queue. Note that for the \ac{TT} traffic model, all devices with the same offset are synchronized together and become active at the same time slot. Furthermore, two devices with different offsets may become simultaneously active in case of retransmission, where the probability of simultaneous activity depends on the relative offset values between the devices and the decoding threshold $\theta$.

\subsection{Age of Information}\label{sec:AoI}
\ac{AoI} quantifies the freshness (i.e., timeliness) of information transmitted by the devices within the network \cite{Kaul2012}. For any link within the considered time slotted system, the metric $\Delta_o(t)$ tracks the \ac{AoI} evolution with time as shown in Fig. \ref{fig:AoI}. Assume that the $o$-th packet is generated at time $Y_o(t)$, then $\Delta_o(t+1)$ is computed recursively as
\begin{equation}
	\Delta_o(t+1)=\left\{\begin{array}{ll}{\Delta_o(t)+1,} & {\text { transmission failure, }} \\ {t-Y_o(t)+1,} & {\text { otherwise }}\end{array}\right.
\end{equation}
Through this paper, we consider the peak \ac{AoI}, termed through the subsequent sections \ac{PAoI}, which is defined as the value of age resulted immediately prior to receiving the $i$-th update \cite{Huang2015}. The increased focus on the \ac{PAoI} stems from the guaranteed system performance insights it unveils. In addition, the minimization of the \ac{PAoI} may be required for time critical applications.  \cite{Xu2019}. To this end, conditioned on a fixed, yet generic spatial realization, the spatially averaged \ac{PAoI},\footnotetext{It is noteworthy to mention that the considered PAoI in this work incorporates temporal and spatial averaing.} as observed from Fig. \ref{fig:AoI}, is computed as
\begin{equation}\label{eq:peakAoI}
	\mathbb{E}\{\Delta_p|  \mathrm{\Phi}\} = \mathbb{E}^{!}\Big\{\mathcal{I}_o + \mathcal{W}_o |  \mathrm{\Phi}, \mathrm{\Psi} \Big\},
\end{equation}
where $\mathbb{E}^{!}\{.\}$ is the reduced Palm expectation \cite{Haenggi2012}, $\mathcal{I}_o$ and $\mathcal{W}_o$ denote the inter-arrival time between consecutive packets and the waiting time of a generic packet in the queue, respectively. As observed, the evaluation of the waiting time is required to evaluate the \ac{PAoI}. The waiting time depends on, among other parameters, the adopted traffic model, queue distribution and network-wide aggregate interference. Throughout this paper, we provide a spatiotemporal mathematical framework to characterize the \ac{PAoI}.
\begin{figure}
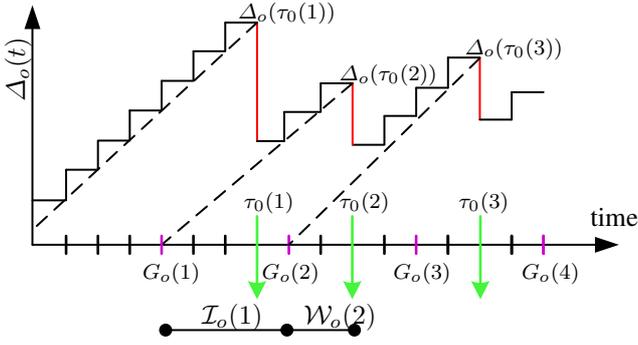

	\begin{center}
		\ifCLASSOPTIONdraftcls
		\input{system_model/figures/AoI/1col/AoI.tex}
		\else
		\input{system_model/figures/AoI/2col/AoI.tex}
		\fi			
		\caption{AoI evolution of a typical link. The time stamps $G_o(n)$ and $\tau_o(n)$ denote the time at which the $n$-th packet was generated and successfully delivered. $\mathcal{I}_o(1)$ and $\mathcal{W}_o(2)$ denote the inter-arrival time and the waiting times.}
		\label{fig:AoI}
	\end{center}
\end{figure}
	\thispagestyle{empty}
	\section{Macroscopic Large Scale Analysis}\label{sec:sg_analysis}
Through this section, the network-wide aggregate interference will be discussed for the \ac{TT} and \ac{ET} traffic models.  First we start with the \ac{TT} in Section \ref{SecTT}. Afterwards, the ET traffic is presented in Section \ref{SecET}.
\subsection{TT Traffic}\label{SecTT}

Due to uplink association, the devices \ac{PP} $\mathrm{\Phi}$ is a Poisson-Voronoi perturbed \acp{PP} with intensity ${\lambda}$ \cite{Singh2015, ElSawy_meta, Vega2016, Renzo2016}. The periodic TT traffic can be incorporated to the devices \ac{PP} via the notion of marked \acp{PP}. That is, let $\tilde{\mathrm{\Phi}}=\{x_i,\beta_i\}$ be a marked \ac{PP} with points $x_i\in \mathrm{\Phi}$ and time offset marks $\beta_i$ drawn from the uniform distribution $\mathbb{P}\{\beta_i=\tau\}=1/T\; ,\tau\in\{0,1,\cdots,T-1\}$. In addition, let $\tilde{\mathrm{\Phi}}_\tau=\{(x_i,\beta_i) \in\tilde{\mathrm{\Phi}} : \beta_i=\tau\}$ be the \ac{PP} where all the devices have identical time offset. Due to the independent and uniform distribution of the time offsets, the intensity of $\tilde{\mathrm{\Phi}}_{\tau}$ for each $\tau \in \{0,1,\cdots,T-1\}$ is  $\frac{\lambda}{T}$.  Note that, all the devices within the same $\tilde{\mathrm{\Phi}}_\tau$ have synchronized packet generation every $T$ time slots, and hence, always interfere together in their first transmission attempt. On the other hand, two devices within different sets $\tilde{\mathrm{\Phi}}_{\tau_1}$ and $\tilde{\mathrm{\Phi}}_{\tau_2}$ for $\tau_1 \neq \tau_2$ may only interfere together due to retransmissions.   A pictorial illustrations of the transmission and mutual interference of four devices in the TT traffic model is shown in Fig.~\ref{fig:interf_per}.

Focusing on a fixed, yet arbitrary, spatial realization of $\mathrm{\Psi}$ and $\tilde{\mathrm{\Phi}}$, let $(u_o,\beta_o) \in \tilde{\mathrm{\Phi}}$,  $b_o = \text{argmin}_{b\in \mathrm{\Psi}} ||u_o-b||$ and $r_o=||u_o-b_o||$ define, respectively, the location, time offset, serving BS, and association distance of a randomly selected $o$-th device, where $||.||$ is the Euclidean norm. For the ease of notation, we define the set $\tilde{\mathrm{\Phi}}_{o,\kappa}=\{r_i=||x_i-b_o|| :  (x_i, \beta_i) \in \tilde{\mathrm{\Phi}}, \beta_i =\kappa\}$ that contains the relative distances to the serving BS of the $o$-th device from all devices, with time offset $\beta=\kappa$. Due to the adopted TT packet generation and persistent transmission scheme, the SIR exhibit a regular time slot dependent pattern that is repeated every $T$ time slots. In particular, let $\ell\in \mathbb{Z}$ be an integer and $\tau\in\{0,1,\cdots,T-1\}$ be a generic time slot within the duty cycle $T$, then the SIR of the $o$-th device at the $(\tau+\ell T)$-th time slot is given by
\begin{equation}\label{eq:SIR_TT}
\text{SIR}_{o,\tau+\ell T}^T \!=\! \frac{\rho h_{o}r_o^{\eta(1-\epsilon)}}{\!\!\!\!\!\underbrace{\sum_{r_i\in \tilde{\mathrm{\Phi}}_{o,\tau}} \!\!\!\!\! P_{i} g_{i}r_i^{-\eta}}_{\text{determistic for each $\tau$}} \!\!\! + \!\!\! \underbrace{\sum_{\kappa \neq \tau} \sum_{r_m\in \tilde{\mathrm{\Phi}}_{o,\kappa}}\!\!\!\!\!\! \mathbbm{1}_{\{a^{(m)}_{\kappa}(\tau+\ell T)\}} P_{m} g_{m}r_m^{-\eta}}_{\text{ probabilistic retransmissions}}}, 
\end{equation}
where $h_o$ is the intended channel power gain, $a^{(m)}_{\kappa}(\tau+ \ell T)$ is the event that the $m$-th device with offset $\beta_m=\kappa$ has a non-empty queue at the $(\tau+\ell T)$-th time slot, $\mathbbm{1}_{\{\cdot\}}$ is an indicator function that is equal to 1 if the event ${\{\cdot\}}$ is true and zero otherwise, $P_i$ ($P_m$) and $g_i$ ($g_m$) denote the $i$-th ($m$-th)  uplink transmit power and its channel power gain, respectively. 
 
Let $p^{(m)}_{\kappa,\tau} = \mathbb{E}\left\{\mathbbm{1}_{\{a^{(m)}_{\kappa}(\tau+ \ell T)\}} | m,\kappa,\tau\right\}$ be the probability that the $m$-th device with time offset $\beta_m=\kappa$ has a non-empty queue at the $\tau$-th time slot within any cycle $\tau+\ell T$. Then the intensity of the interfering devices within the $\tau$-th time slot  is given by
\begin{equation} \label{eq:inten_TT}
\lambda_\tau =  \frac{(1+\varTheta_\tau)\lambda}{T},	
\end{equation}
where  $\varTheta_\tau \in [0,T-1]$ is given by $\varTheta_\tau= \sum_{\kappa \neq \tau} \mathbb{E}_{\tilde{\mathrm{\Phi}}_{\kappa}}\{p^{(m)}_{\kappa,\tau}\}$. Recalling that the intensity of the devices with each distinct time offset is $\frac{\lambda}{T}$, it is clear that $\varTheta_\tau $ depicts the aggregate percentiles of devices with time offsets $\kappa\neq \tau$ that are active at time slot $\tau$. At the extreme case of flawless transmissions,  $\varTheta_\tau = 0$ and $\lambda_\tau=\frac{\lambda}{T}$, where only synchronized devices with newly generated  packets mutually interfere together. On the other extreme, assuming  backlogged queues due to poor transmission success probabilities, $\varTheta_\tau =T-1$, and hence, $\lambda_\tau=\lambda$, where all devices are always active and mutually interfere together. In realistic cases,  $0 \leq \varTheta_\tau \leq T-1$, which is the focus of the current analysis. 

\begin{figure*}[t!]
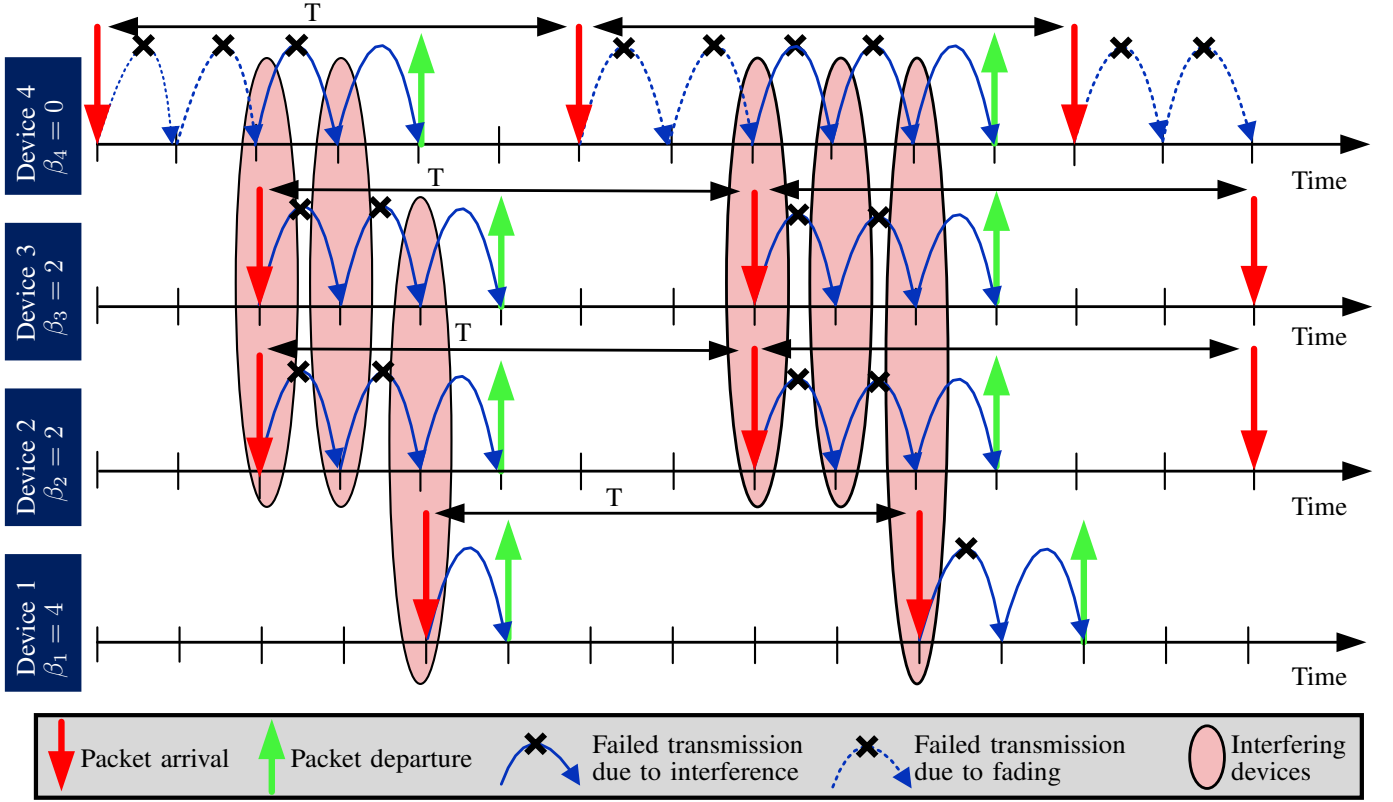

	\centering
	\ifCLASSOPTIONdraftcls
	\input{QT_anaylsis/figures/interf_per/1col/interf_per.tex}
	\vspace{-0.4in}
	\else
	\input{QT_anaylsis/figures/interf_per/2col/interf_per.tex}
	\fi 		
	\caption{Packets generation, departure process, and mutual interference between devices under TT traffic  with duty cycle $T=6$.} 
	\label{fig:interf_per}
\end{figure*}

As mentioned earlier, devices are only active when they have non-empty queues. A packet at the queue of a generic $o$-th device departs from its queue in the time slot $\tau\in \{0,1,\cdots,T-1\}$ if $\mathbb{P}\{\text{SIR}_{o,\tau}^T >\theta\}$. Since a packet is generated every $T$ slots, it is required that $\text{SIR}_{o,\tau}^T$ exceeds the threshold $\theta$ at least once for any of the time slots  $\tau \in \{0,1, \cdots,T-1\}$. Once the generated packet departs and the queue is empty, the device remains idle for the rest of the cycle until the next packet generation (cf. Fig.~\ref{fig:interf_per}). Otherwise, the departure rate is not sufficient to cope with the periodic  packet generation and packets keep accumulating in the device's queue. Such devices are never idle and are denoted hereafter as unstable devices.

As illustrated from \eqref{eq:SIR_TT} and \eqref{eq:inten_TT}, the activities of interfering devices, and consequently,  $\text{SIR}_{o,\tau}^T$ are location and time slot dependent. Due to the fixed realization of the network, the static time offsets, and the periodic generation of packets, each device experiences a location and slot-dependent pattern of $\text{SIR}_{o,\tau+\ell T}^T$ for $\tau \in \{0,1, \cdots,T-1\}$ that is repeated every cycle $\ell T,\;\forall \ell\in \mathbb{Z}$. Despite the randomness in the channel gains and the probabilistic interference of devices with different offsets, the network geometry and the periodic packet generation with static offsets have the dominating effect that highly correlates $\text{SIR}_{o,\tau+\ell T}^T$ for each $\tau$ across different cycles. Such location and time slot dependence of the SIR yields intractable analysis. Furthermore, there is no known tractable exact analysis for Poisson-Voronoi perturbed \acp{PP} \cite{Singh2015, ElSawy_meta, Vega2016, Renzo2016}. Hence, for th sake of analytical tractability, we resort to the following two approximations. 

\begin{approximation} \label{approx_DetInterf}
	The location and time slot dependent transmission success probabilities $\mathbb{P}\{\text{SIR}_{o,\tau}^T >\theta\}$ of the BSs in ${\mathrm{\Psi}}$ and  devices in $\tilde{\mathrm{\Phi}}$ are approximated by the location-dependent transmission success probabilities $\mathbb{P}\{\hat{\text{SIR}}_{o}^T >\theta\}$ where each BS in ${\mathrm{\Psi}}$ sees a fixed panorama of always active interfering devices  constituting a fixed, yet arbitrary, PPP $\hat{\mathrm{\Phi}}$ with intensity function
	\begin{equation} \label{eq:intTT2}
\lambda_T(x) =  \frac{(1+\varTheta_{T})\lambda}{T}(1-e^{-\pi\lambda x^2}) .
	\end{equation}
	where $0 \leq \varTheta_{T} \leq T-1$ is as defined in \eqref{eq:inten_TT} for a generic time slot.\footnotetext{The subscript T in $\varTheta_{T}$ depicts the TT trafic and not a specific value of the duty cycle $T$.}
\end{approximation}
\begin{remark}  
	Approximation~\ref{approx_DetInterf} can be regarded as approximating the success probability  $\mathbb{P}\{\text{SIR}_{o,\tau}^T >\theta\}$ of each device across different time slots within the same cycle $T$ by an approximate mean value $\mathbb{P}\{\hat{\text{SIR}}_{o}^T >\theta\}\approx \mathbb{E}_{\tau}\{\mathbb{P}\{{\text{SIR}}_{o},\tau\}^T >\theta\}$ to alleviate the time-slot dependence. The approximating PPP $\hat{\mathrm{\Phi}}$  is assumed to be static to account for the temporal correlations between different cycles, and hence, capture the location dependent performance of the devices. Note that the intensity function in \eqref{eq:intTT2} is sensitive to the effect of unsaturated TT traffic through the parameter $\frac{(1+\varTheta_{T})\lambda}{T}$. Furthermore, \eqref{eq:intTT2} is also sensitive to the uplink association through the factor $(1-e^{-\pi\lambda x^2})$ \cite{Singh2015, ElSawy_meta, Vega2016, Renzo2016}.  It is worth noting that the validity of such approximation is validated via independent Monte-Carlo simulations in Section \ref{sec:simulation_results}.
\end{remark}

\begin{approximation} \label{approxTT}
	The transmission powers of the interfering devices are uncorrelated.
\end{approximation}
\begin{remark}  
	Approximation~\ref{approxTT} ignores the correlations among the sizes of adjacent Voronoi cells, which lead to correlated transmission powers of devices due to the adopted fractional path-loss inversion power control scheme. Such approximation is widely utilized in the literature to maintain mathematical tractability \cite{Singh2015, ElSawy_meta, Vega2016, Renzo2016}. We further validate Approximation~\ref{approxTT}  via independent Monte-Carlo simulations in Section \ref{sec:simulation_results}.
\end{remark}
By virtue of Approximation~\ref{approx_DetInterf}, the time slot indices $\kappa$ and $\tau$ are dropped hereafter. Furthermore, exploiting Approximations~\ref{approx_DetInterf} and \ref{approxTT} along with the mapping and displacement theorems of the PPP~\cite{Haenggi2012}, the effect of the power control and path-loss can be incorporated to the intensity function of the approximating PPP. That is, the PPP of the interfering devices $\hat{\mathrm{\Phi}}$ can be mapped to a 1-D \ac{PPP} with unit transmission powers and inverse linear path-loss function. After mapping and displacement, following  \cite[Lemma 2]{ElSawy_meta}, the intensity function in \eqref{eq:intTT2} becomes  
\begin{equation} \label{eq:TT_int3}
\tilde{\lambda}_T(\omega) =\frac{2(1+\varTheta_{T})(\pi\lambda)^{1-\epsilon}\rho^{\frac{2}{\eta}}}{T\eta \omega^{1-\frac{2}{\eta}}} \gamma\Big(1+\epsilon, \pi\lambda(\omega\rho)^{\frac{2}{\eta(1-\epsilon)}}\Big).
\end{equation}
Using the intensity function in \eqref{eq:TT_int3}, the transmission success probability in the TT traffic model is defined as
\begin{align}\label{eq:PS}
P_s(\theta) &= \mathbb{P}^{ !}\left\{\text{SIR}_{o}^T>\theta | \hat{\mathrm{\Phi}}, \mathrm{\Psi} \right\}, \nonumber \\
 &= \prod_{\omega_i\in \tilde{\mathrm{\Phi}}_{T}}  \mathbb{E}^{!} \Big[ \Big(\frac{1}{1 + \frac{\theta r_o^{\eta(1-\epsilon)}}{\rho \omega_i}} \Big)\Big| \hat{\mathrm{\Phi}}, \mathrm{\Psi} \Big],
\end{align}
where $\hat{\mathrm{\Phi}}_{T} =\{ \omega_ i = \frac{r_i}{P_i},\; \forall r_i \in \hat{\mathrm{\Phi}}_o\}$, and the set $\hat{\mathrm{\Phi}}_o$ contains all relative distances from the approximating PPP $\hat{\mathrm{\Phi}}$ to the serving BS of the $o$-th device. The computation in \eqref{eq:PS} follows from the exponential distribution of $h_o$ and $h_i$. To account for the location dependent success probability, $P_s(\theta)$ is considered as a random variable across different devices. The meta distribution of the success probability models such variation across the network \cite{Haenggi_meta, Haenggi_meta2} as 
\begin{equation}
\bar{F}(\theta, \xi)=\mathbb{P}^{!}\{P_{s}(\theta)>\xi | \hat{\mathrm{\Phi}}_{o}, \mathrm{\Psi}\},
\end{equation}
where $\xi$ denotes the percentile of devices within the network that achieves an \ac{SIR} equals to $\theta$. Following \cite{ElSawy_meta, Haenggi_meta}, the meta distribution of the success probability for the \ac{TT} traffic $F_T(\theta, \xi)$ is be approximated as
\begin{align}\label{eq:metaPDF}
F_T(\theta, \xi) &\approx I_{\xi}\left(\frac{M_{1,T}	\hat{M}_T}{\left(M_{2,T}-M_{1,T}^{2}\right)}, \frac{\left(1-M_{1,T}\right)	\hat{M}_T}{\left(M_{2,T}-M_{1,T}^{2}\right)}\right), \nonumber \\
\hat{M}_T &= M_{1,T}-M_{2,T},
\end{align}
where $I_{\xi}(a, b)=\int_{0}^{\xi} t^{a-1}(1-t)^{b-1} \mathrm{d} t$ is the regularized
incomplete beta function, $M_{1,T}$ and $M_{2,T}$ are the first and second moments of $P_s(\theta)$ under the \ac{TT} traffic model. The approximate moments of $P_s(\theta)$ for the \ac{TT} traffic model are given via the following lemma.
\begin{lemma}\label{lem:meta_lemma}
	The moments of the transmission success probabilities in uplink network under TT traffic with duty cycle $T$ are approximated by ${M}_{b,T} \sim \tilde{M}_{b,T}$ as given in (\ref{eq:lemma1}), where $\gamma(a,y) = \int_{0}^{b} t^{a-1}\text{e}^{-t}dt$ is the lower incomplete gamma function. 
\end{lemma}
\begin{proof}
	See Appendix \ref{se:Appendix_A}
\end{proof}

\begin{figure*}[!t]
	\normalsize
	\ifCLASSOPTIONdraftcls
	\begin{align}\label{eq:lemma1}
	\tilde{M}_{b,T} &= \int_{0}^{\infty} \exp \Bigg\{-z- \Big(\frac{\big(1+\varTheta_{T}\big)2z^{1-\epsilon}}{T\eta } \nonumber \\
	&\int_{\mathbbm{1}\{\epsilon = 1\}}^{\infty} y^{\frac{2}{\eta}-1}\big( 1-\Big(\frac{y}{y + \theta}\Big)^b\big) \gamma\Big(1+\epsilon, zy^{\frac{2}{\eta(1-\epsilon)}}\Big) dy\Big)\Bigg\} dz. 
	\end{align}
	\else
	\begin{align}\label{eq:lemma1}
	\tilde{M}_{b,T} &= \int_{0}^{\infty} \exp \Bigg\{-z- \Big(\frac{\big(1+\varTheta_{T}\big)2z^{1-\epsilon}}{T\eta } \int_{\mathbbm{1}\{\epsilon = 1\}}^{\infty} y^{\frac{2}{\eta}-1}\big( 1-\Big(\frac{y}{y + \theta}\Big)^b\big) \gamma\Big(1+\epsilon, zy^{\frac{2}{\eta(1-\epsilon)}}\Big) dy\Big)\Bigg\} dz.
	\end{align}
	\fi
	\hrulefill
\end{figure*}

\subsubsection{\textbf{Network Categorization}}\label{QosClass}
It is observed from Lemma \ref{lem:meta_lemma} that the macroscopic network-wide aggregate characterization depends on the parameter $\varTheta_{T}$. Before delving into the details of such characterization, we first discretize the meta distribution of $P_s(\theta)$ through uniform network partitioning \cite{Chisci2019}. Categorizing each devices within the network into a distinctive \ac{QoS} class is not feasible due to the continuous support of $P_s(\theta) \in[0,1]$. Consequently, the transmission success probability is quantized into $N$ \ac{QoS} classes.\footnote{The continuous random variable $P_s(\theta)$ with distribution $ f_{P_s}(\omega)$ is quantized to an equally-probable uniform random variable $\mathbf{d} = [d_1 \; d_2 \; \cdots \; d_N]$.} The network categorization process of the distribution in (\ref{eq:metaPDF}) for the $n$-th class is conducted as follows
\begin{equation}\label{eq:pmf1}
F_{P_s}(\omega_n)-F_{P_s}(\omega_{n+1})=\int_{\omega_n}^{\omega_{n+1}} f_{P_s}(\omega) d\omega=\frac{1}{N},
\end{equation}
where $n\in \{1,2,\cdots N\}$. Afterwards, the discrete probability mass function $d_n$ (i.e., $F_{P_s}(d_n) = \frac{1}{N}$) can be evaluated using the bisection method as 
\begin{equation}\label{eq:pmf2}
\int_{\omega_n}^{d_n} f_{P_s}(\omega)d\omega = \int_{d_n}^{\omega_{n+1}} f_{P_s}(\omega)d\omega.
\end{equation} 
The computation of $d_n,\; \forall n$ via (\ref{eq:pmf1}) and (\ref{eq:pmf2}) quantizes the meta distribution of $P_s(\theta)$ into $N$ equiprobable classes as shown in Fig. \ref{fig:quant_meta}. The queue's departure rate of a device belonging to the $n$-class is determined by  $d_n$. Now we are in position to characterize the TT traffic parameter  $\varTheta_{T}$. 
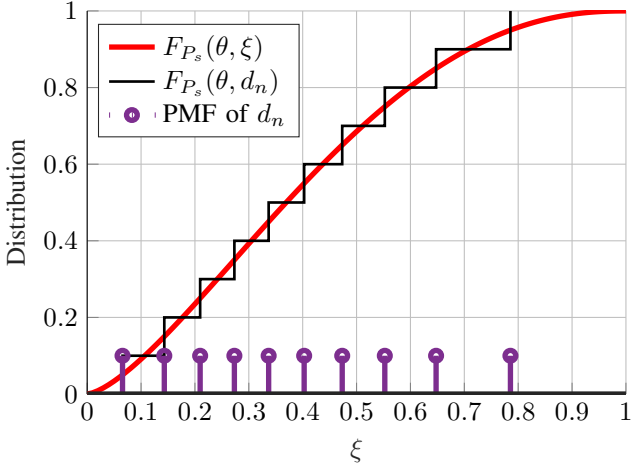
\begin{figure} 
	\centering
	\ifCLASSOPTIONdraftcls
	\definecolor{mycolor1}{rgb}{0.49020,0.18039,0.56078}%

\begin{tikzpicture}

\begin{axis}[%
width=0.6\columnwidth,
height=1.5in,
scale only axis,
xmin=0,
xmax=1,
xlabel style={font=\color{white!15!black}},
xlabel={$\xi$},
ymin=0,
ymax=1,
ytick={0,0.2,0.4,0.6,0.8,1},
xtick={0,0.1,0.2,0.3,0.4,0.5,0.6,0.7,0.8,0.9,1},
ylabel style={font=\color{white!15!black}},
ylabel={Distribution},
axis x line*=bottom,
axis y line*=left,
xmajorgrids,
ymajorgrids,
legend style={at={(0.02,0.55)}, anchor=south west, legend cell align=left, align=left, draw=white!15!black}
]
\addplot [color=red, line width=2.0pt]
table[row sep=crcr]{%
	0	0\\
	0.01	0.00316147369524374\\
	0.02	0.00882172754478877\\
	0.03	0.0160243637660381\\
	0.04	0.0244154579543246\\
	0.05	0.0337837571741215\\
	0.06	0.043982309739499\\
	0.07	0.0549000574866723\\
	0.08	0.0664485383584482\\
	0.09	0.0785546785910195\\
	0.1	0.0911564860060308\\
	0.11	0.104200290013568\\
	0.12	0.117638876499696\\
	0.13	0.131430173837579\\
	0.14	0.145536295509061\\
	0.15	0.159922822930666\\
	0.16	0.174558255558411\\
	0.17	0.18941358080871\\
	0.18	0.204461931890115\\
	0.19	0.219678311495064\\
	0.2	0.235039365739724\\
	0.21	0.250523197062651\\
	0.22	0.266109207764718\\
	0.23	0.28177796795935\\
	0.24	0.297511103195023\\
	0.25	0.313291198098395\\
	0.26	0.329101713189171\\
	0.27	0.344926912619295\\
	0.28	0.360751801045421\\
	0.29	0.376562068193998\\
	0.3	0.39234403995015\\
	0.31	0.40808463501454\\
	0.32	0.423771326340905\\
	0.33	0.439392106701285\\
	0.34	0.454935457833978\\
	0.35	0.470390322716766\\
	0.36	0.485746080579267\\
	0.37	0.500992524326858\\
	0.38	0.516119840096992\\
	0.39	0.531118588708916\\
	0.4	0.54597968880141\\
	0.41	0.560694401481357\\
	0.42	0.575254316329802\\
	0.43	0.589651338632364\\
	0.44	0.603877677718117\\
	0.45	0.617925836305845\\
	0.46	0.631788600769273\\
	0.47	0.645459032243929\\
	0.48	0.65893045850784\\
	0.49	0.672196466576648\\
	0.5	0.685250895961113\\
	0.51	0.698087832541443\\
	0.52	0.710701603018723\\
	0.53	0.723086769908891\\
	0.54	0.73523812704939\\
	0.55	0.74715069559294\\
	0.56	0.758819720466763\\
	0.57	0.770240667279302\\
	0.58	0.781409219659894\\
	0.59	0.792321277020183\\
	0.6	0.802972952729238\\
	0.61	0.813360572697503\\
	0.62	0.82348067436786\\
	0.63	0.833330006115298\\
	0.64	0.84290552706002\\
	0.65	0.852204407302355\\
	0.66	0.861224028591577\\
	0.67	0.869961985444899\\
	0.68	0.878416086737417\\
	0.69	0.886584357788906\\
	0.7	0.894465042979161\\
	0.71	0.902056608930233\\
	0.72	0.909357748301663\\
	0.73	0.91636738425389\\
	0.74	0.923084675645824\\
	0.75	0.929509023045427\\
	0.76	0.9356400756478\\
	0.77	0.941477739214249\\
	0.78	0.947022185169238\\
	0.79	0.952273861021269\\
	0.8	0.957233502310285\\
	0.81	0.961902146330657\\
	0.82	0.96628114793847\\
	0.83	0.970372197829431\\
	0.84	0.974177343776026\\
	0.85	0.977699015449496\\
	0.86	0.980940053638466\\
	0.87	0.98390374493423\\
	0.88	0.986593863317886\\
	0.89	0.989014720613508\\
	0.9	0.99117122855845\\
	0.91	0.993068976449822\\
	0.92	0.994714330249099\\
	0.93	0.996114562224595\\
	0.94	0.997278025818323\\
	0.95	0.998214400924384\\
	0.96	0.998935056169969\\
	0.97	0.999453623675835\\
	0.98	0.999787013265266\\
	0.99	0.999957554480932\\
	1	1\\
};
\addlegendentry{$F_{P_s}(\theta, \xi)$}

\addplot [color=black, line width=1.0pt]
table[row sep=crcr]{%
	-inf	0\\
	0.0655875205993652	0\\
	0.0655875205993652	0.1\\
	0.14312219619751	0.1\\
	0.14312219619751	0.2\\
	0.209662914276123	0.2\\
	0.209662914276123	0.3\\
	0.273205280303955	0.3\\
	0.273205280303955	0.4\\
	0.336818218231201	0.4\\
	0.336818218231201	0.5\\
	0.402720928192139	0.5\\
	0.402720928192139	0.6\\
	0.473352909088135	0.6\\
	0.473352909088135	0.7\\
	0.552421092987061	0.7\\
	0.552421092987061	0.8\\
	0.647600650787354	0.8\\
	0.647600650787354	0.9\\
	0.785600185394287	0.9\\
	0.785600185394287	1\\
	inf	1\\
};
\addlegendentry{$F_{P_s}(\theta, d_n)$}

\addplot[ycomb, color=mycolor1, line width=2.0pt, mark=o, mark options={solid, mycolor1}] table[row sep=crcr] {%
	0.0655875205993652	0.1\\
	0.14312219619751	0.1\\
	0.209662914276123	0.1\\
	0.273205280303955	0.1\\
	0.336818218231201	0.1\\
	0.402720928192139	0.1\\
	0.473352909088135	0.1\\
	0.552421092987061	0.1\\
	0.647600650787354	0.1\\
	0.785600185394287	0.1\\
};
\addplot[forget plot, color=white!15!black, line width=2.0pt] table[row sep=crcr] {%
	0	0\\
	1	0\\
};
\addlegendentry{$\text{PMF of }d_n$}

\end{axis}
\end{tikzpicture}%
	\else
	\definecolor{mycolor1}{rgb}{0.49020,0.18039,0.56078}%

\begin{tikzpicture}

\begin{axis}[%
width=0.8\columnwidth,
height=2in,
scale only axis,
xmin=0,
xmax=1,
xlabel style={font=\color{white!15!black}},
xlabel={$\xi$},
ymin=0,
ymax=1,
ytick={0,0.2,0.4,0.6,0.8,1},
xtick={0,0.1,0.2,0.3,0.4,0.5,0.6,0.7,0.8,0.9,1},
ylabel style={font=\color{white!15!black}},
ylabel={Distribution},
axis x line*=bottom,
axis y line*=left,
xmajorgrids,
ymajorgrids,
legend style={at={(0.02,0.67)}, anchor=south west, legend cell align=left, align=left, draw=white!15!black}
]
\addplot [color=red, line width=2.0pt]
table[row sep=crcr]{%
	0	0\\
	0.01	0.00316147369524374\\
	0.02	0.00882172754478877\\
	0.03	0.0160243637660381\\
	0.04	0.0244154579543246\\
	0.05	0.0337837571741215\\
	0.06	0.043982309739499\\
	0.07	0.0549000574866723\\
	0.08	0.0664485383584482\\
	0.09	0.0785546785910195\\
	0.1	0.0911564860060308\\
	0.11	0.104200290013568\\
	0.12	0.117638876499696\\
	0.13	0.131430173837579\\
	0.14	0.145536295509061\\
	0.15	0.159922822930666\\
	0.16	0.174558255558411\\
	0.17	0.18941358080871\\
	0.18	0.204461931890115\\
	0.19	0.219678311495064\\
	0.2	0.235039365739724\\
	0.21	0.250523197062651\\
	0.22	0.266109207764718\\
	0.23	0.28177796795935\\
	0.24	0.297511103195023\\
	0.25	0.313291198098395\\
	0.26	0.329101713189171\\
	0.27	0.344926912619295\\
	0.28	0.360751801045421\\
	0.29	0.376562068193998\\
	0.3	0.39234403995015\\
	0.31	0.40808463501454\\
	0.32	0.423771326340905\\
	0.33	0.439392106701285\\
	0.34	0.454935457833978\\
	0.35	0.470390322716766\\
	0.36	0.485746080579267\\
	0.37	0.500992524326858\\
	0.38	0.516119840096992\\
	0.39	0.531118588708916\\
	0.4	0.54597968880141\\
	0.41	0.560694401481357\\
	0.42	0.575254316329802\\
	0.43	0.589651338632364\\
	0.44	0.603877677718117\\
	0.45	0.617925836305845\\
	0.46	0.631788600769273\\
	0.47	0.645459032243929\\
	0.48	0.65893045850784\\
	0.49	0.672196466576648\\
	0.5	0.685250895961113\\
	0.51	0.698087832541443\\
	0.52	0.710701603018723\\
	0.53	0.723086769908891\\
	0.54	0.73523812704939\\
	0.55	0.74715069559294\\
	0.56	0.758819720466763\\
	0.57	0.770240667279302\\
	0.58	0.781409219659894\\
	0.59	0.792321277020183\\
	0.6	0.802972952729238\\
	0.61	0.813360572697503\\
	0.62	0.82348067436786\\
	0.63	0.833330006115298\\
	0.64	0.84290552706002\\
	0.65	0.852204407302355\\
	0.66	0.861224028591577\\
	0.67	0.869961985444899\\
	0.68	0.878416086737417\\
	0.69	0.886584357788906\\
	0.7	0.894465042979161\\
	0.71	0.902056608930233\\
	0.72	0.909357748301663\\
	0.73	0.91636738425389\\
	0.74	0.923084675645824\\
	0.75	0.929509023045427\\
	0.76	0.9356400756478\\
	0.77	0.941477739214249\\
	0.78	0.947022185169238\\
	0.79	0.952273861021269\\
	0.8	0.957233502310285\\
	0.81	0.961902146330657\\
	0.82	0.96628114793847\\
	0.83	0.970372197829431\\
	0.84	0.974177343776026\\
	0.85	0.977699015449496\\
	0.86	0.980940053638466\\
	0.87	0.98390374493423\\
	0.88	0.986593863317886\\
	0.89	0.989014720613508\\
	0.9	0.99117122855845\\
	0.91	0.993068976449822\\
	0.92	0.994714330249099\\
	0.93	0.996114562224595\\
	0.94	0.997278025818323\\
	0.95	0.998214400924384\\
	0.96	0.998935056169969\\
	0.97	0.999453623675835\\
	0.98	0.999787013265266\\
	0.99	0.999957554480932\\
	1	1\\
};
\addlegendentry{$F_{P_s}(\theta, \xi)$}

\addplot [color=black, line width=1.0pt]
table[row sep=crcr]{%
	-inf	0\\
	0.0655875205993652	0\\
	0.0655875205993652	0.1\\
	0.14312219619751	0.1\\
	0.14312219619751	0.2\\
	0.209662914276123	0.2\\
	0.209662914276123	0.3\\
	0.273205280303955	0.3\\
	0.273205280303955	0.4\\
	0.336818218231201	0.4\\
	0.336818218231201	0.5\\
	0.402720928192139	0.5\\
	0.402720928192139	0.6\\
	0.473352909088135	0.6\\
	0.473352909088135	0.7\\
	0.552421092987061	0.7\\
	0.552421092987061	0.8\\
	0.647600650787354	0.8\\
	0.647600650787354	0.9\\
	0.785600185394287	0.9\\
	0.785600185394287	1\\
	inf	1\\
};
\addlegendentry{$F_{P_s}(\theta, d_n)$}

\addplot[ycomb, color=mycolor1, line width=2.0pt, mark=o, mark options={solid, mycolor1}] table[row sep=crcr] {%
	0.0655875205993652	0.1\\
	0.14312219619751	0.1\\
	0.209662914276123	0.1\\
	0.273205280303955	0.1\\
	0.336818218231201	0.1\\
	0.402720928192139	0.1\\
	0.473352909088135	0.1\\
	0.552421092987061	0.1\\
	0.647600650787354	0.1\\
	0.785600185394287	0.1\\
};
\addplot[forget plot, color=white!15!black, line width=2.0pt] table[row sep=crcr] {%
	0	0\\
	1	0\\
};
\addlegendentry{$\text{PMF of }d_n$}

\end{axis}
\end{tikzpicture}%
	\fi 	
	\caption{Quantized meta distribution for $N=10$, hypothetical $\varTheta_{T} = 0.5$ and $\theta = 5$ dB.}
	\label{fig:quant_meta} 
\end{figure} 

\subsubsection{\textbf{$\mathbf{\varTheta_{T}}$ Characterization for TT Traffic}}

As mentioned earlier, for a given set of synchronized devices (i.e., with equal time offset), $\varTheta_{T}$ depicts the aggregate percentiles of retransmitting devices from all other distinct time offsets. The first step to characterize $\varTheta_{T}$ is the determine the set of always active devices, if any. Particularly, a \ac{QoS} class that impose a departure rate less than the packet arrival rate yield always active devices that continuously interfere with other devices irrespective of their relative time offsets. Stable and unstable \ac{QoS} classes are discriminated via a transmission success probability threshold equal to the packet arrival probability \cite{Loynes1962}.  Consequently,  $\mathcal{S} =\{d_n \geq \frac{1}{T-1} \;|\; n \in\{1,2,\cdots,N\} \}$ is the set of stable \ac{QoS} classes (i.e., devices belonging to this class can empty its queue within the duty cycle $T$). Visually, Fig. \ref{fig:interf_per} depicts an example scenario with three devices, each belong to a given \ac{QoS} class. Device 1, belonging to the lowest performing class (i.e., one with lowest $d_n$), requires more time slots to successfully transmit its packet. It is noteworthy to mention that transmission failures might occur due to the mutual interference between active devices or fading and path-loss effect. In accordance,  $\mathcal{U} =\{d_n < \frac{1}{T-1} \;|\; n \in\{1,2,\cdots,N\} \}$ is the set of unstable \ac{QoS} classes. Devices belonging to an unstable \ac{QoS} class are not able to  empty their queues within the packet generation duty cycle $T$. Thus, their queues will have infinite size and become unstable. For mathematical tractability, we adopt the following approximation in our work. 

\begin{approximation} \label{approxIID}
Queues employed at the devices are \ac{QoS}-aware but have  temporally-independent departures.
\end{approximation}
\begin{remark}
	The temporal correlation of interference is captured by the static QoS class of each device. That is, the departure probability of a device belonging to $n$-th QoS class remains $d_n$. Once the QoS class is fixed, the departures from the same device across different time slots are considered to be independent due to the randomness introduced by the channel fading and interfering devices activity profiles.
\end{remark}
Let $r_{\mathcal{S}_n,k}$ be the probability that a device belonging to a stable $n$-th \ac{QoS} is active for $k$-constitutive time slots. Recall that every device has a new generated packet every $T$ time slots and that stable devices, on average, are able to empty their packets within each duty cycle T. Leveraging the temporal independence between the time slots given by Approximation \ref{approxIID}, $r_{\mathcal{S}_n, k} = (1-d_n)^k $. The k-consecutive time slots activity due to transmission failures is illustrated in Fig. \ref{fig:interf_per}. The characterization of the spatially averaged aggregate percentiles of retransmitting devices $\varTheta_{T}$ for the \ac{TT} traffic is given in the following lemma.
\begin{lemma}\label{lem:chi_T}
Consider a TT traffic model with duty cycle T. For each set of synchronized devices, the spatially averaged aggregate percentiles of retransmitting devices from all other distinct time offsets is given by  
	\begin{equation} \label{eq:lemm2}
	\varTheta_{T}=\frac{1}{N}  \sum_{\tau=1}^{T-1} \Big(|\mathcal{U}| +  \sum_{j  = 1}^{|\mathcal{S}|} r_{\mathcal{S}_j,k} \Big),
	\end{equation}
	where $\mathcal{S}=\{d_n \geq \frac{1}{T-1} \;|\; n \in[1,2,\cdots,N] \}$  and $\mathcal{U}=\{d_n < \frac{1}{T-1} \;|\; n \in(1,N) \}$ denote the set of stable and unstable QoS classes, respectively.
\end{lemma}
\begin{proof}
	First, the devices belonging to a \ac{QoS} class that is unstable are always contributing to the aggregate interference. Accordingly, for each distinct time offset, $\frac{|\mathcal{U}|}{N}$ percentiles of the devices will always be interfering every time slot within the window $T$. Second, the set of stable devices with time offset $k$ slots away from a given transmission will only interfere if they have encountered $k$-consecutive transmission failures. Considering all stable QoS classes within each set of devices with distinct time offset, the percentiles of devices that are $k$-slots active can be characterized as $\frac{\sum_{j  = 1}^{|\mathcal{S}|} r_{\mathcal{S}_j,k}}{N}$. Combining the two components (i.e., stable and unstable devices) together and considering all other distinct $T-1$ time offsets within the duty cycle, the lemma is obtained. 
\end{proof}
Iterating through Lemmas \ref{lem:meta_lemma} and \ref{lem:chi_T}, one can evaluate $\Theta_T$ and the meta distribution $\bar{F}_T(\theta, \xi)$. In particular, for any feasible initial value of $\Theta_T$, the moments and the transmission success probabilities for each QoS class can be calculated via \eqref{eq:lemma1}, \eqref{eq:pmf1}, and \eqref{eq:pmf2}. Then, the value of $\Theta_T$ can be updated via~\eqref{eq:lemm2}. Repeating such steps, the aforementioned system of equations  converges to a unique solution by virtue of fixed point theorem \cite{Zhou2016}. After convergence to a unique solution, the waiting time, a generic packet spends in the system till its successful transmission, can be evaluated based on the analysis that will be provided in the next section.

\subsection{ET Traffic}\label{SecET}
Following the same methodology that was presented in the \ac{TT} traffic analysis, the $\text{SIR}_{o,\tau}$ of the $o$-th device at the $\tau$-th time slot  under \ac{ET} traffic is
\begin{equation}\label{eq:SIR_ET}
\text{SIR}_{o,\tau}^E = \frac{P_{o} h_{o}r_o^{\eta(1-\epsilon)}}{\sum_{u_{i} \in \mathrm{\Phi} \backslash u_o} \mathbbm{1}_{\{a_i\}} P_{i} h_{i}r_i^{-\eta}}, 
\end{equation}
where $a_i$ is the event that a generic device has a non-empty queue at steady state. Due to the randomized packet generation and departure, the interference in the ET traffic does not exhibit regular repetitive pattern as in the TT case. Hence, \eqref{eq:SIR_ET} is independent of the time slot index $\tau$, which will be dropped hereafter. 
\begin{figure*}[!t]
	\normalsize
	\ifCLASSOPTIONdraftcls
	\begin{align}\label{eq:lemma2}
	\tilde{M}_{b,E} &= \int_{0}^{\infty} \exp \Bigg\{-z- \frac{2z^{1-\epsilon}}{\eta } \int_{\mathbbm{1}\{\epsilon = 1\}}^{\infty} y^{\frac{2}{\eta}-1}\big( 1-\Big(\frac{y+\theta\varTheta_{E}}{y + \theta}\Big)^b\big) \gamma\Big(1+\epsilon, zy^{\frac{2}{\eta(1-\epsilon)}}\Big) dy\Big)\Bigg\} dz. 
	\end{align}
	\else
	\begin{align}\label{eq:lemma2}
	\tilde{M}_{b,E} &= \int_{0}^{\infty} \exp \Bigg\{-z- \frac{2z^{1-\epsilon}}{\eta } \int_{\mathbbm{1}\{\epsilon = 1\}}^{\infty} y^{\frac{2}{\eta}-1}\big( 1-\Big(\frac{y+\theta\varTheta_{E}}{y + \theta}\Big)^b\big) \gamma\Big(1+\epsilon, zy^{\frac{2}{\eta(1-\epsilon)}}\Big) dy\Big)\Bigg\} dz.
	\end{align}
	\fi
	\hrulefill
\end{figure*}
Analogous to Approximations~\ref{approx_DetInterf} and \ref{approxTT}, let $\hat{\mathrm{\Phi}}_E$ be a PPP with an intensity function $\lambda_E(x) =  \lambda(1-e^{-\pi\lambda x^2})$ that approximates the interference from $\{\mathrm{\Phi}\backslash b_o\}$. Exploiting the mapping and displacement theorems, the interfering PPP seen at a generic BS $b_o \in \mathrm{\Psi}$ can be mapped to a 1-D inhomogeneous PPP $\hat{\mathrm{\Phi}}_{E,o} = \{ s_i = \frac{||x_i-b_o||^{\eta}}{P_i},\; \forall x_i \in \hat{\mathrm{\Phi}}_E \}$ with the following intensity function 
\begin{equation}
\tilde{\lambda}_E(s) =\frac{2(\pi\lambda)^{1-\epsilon}\rho^{\frac{2}{\eta}}}{\eta s^{1-\frac{2}{\eta}}} \gamma\Big(1+\epsilon, \pi\lambda(s\rho)^{\frac{2}{\eta(1-\epsilon)}}\Big).
\end{equation}
Hence, the transmission success  probability for the ET traffic model  is expressed as 
\begin{align} \label{eq:PSET}
P_s(\theta) &= \prod_{s_i\in \tilde{\mathrm{\Phi}}_{E}}  \mathbb{E}^{!} \Big[ \Big(\frac{\bar{\varTheta}_{E}}{1 + \frac{a_i\theta r_o^{\eta(1-\epsilon)}}{\rho s_i}} + \varTheta_{E} \Big)\Big| \mathrm{\Phi}, \mathrm{\Psi} \Big],
\end{align}
where $\varTheta_{E}$ denotes the spatially averaged idle probability (i.e., the probability that a device has an empty queue) under the \ac{ET} traffic at steady state. Different from its TT counterpart in \eqref{eq:PS}, the transmission success probability for the ET in \eqref{eq:PSET} depicts the varying set of interfering devices through the probability of empty queue $\Theta_E$~\cite{Haenggi_meta, ElSawy_meta}. In particular, the higher probability of empty queues (i.e., higher value of $\Theta_E$), the less correlated interference across time slots, and vice versa.

The approximations of $P_s(\theta)$ moments for the \ac{ET} traffic model are given via the following lemma.
\begin{lemma}\label{lem:meta_lemma_2}
	The moments of the transmission success  probabilities in uplink network with ET traffic model with arrival probability $\alpha$ are approximated by $\tilde{M}_{b,E}$ given in (\ref{eq:lemma2}), where $\varTheta_{E}$ is the spatially averaged idle probability. 
\end{lemma}
\begin{proof}
	The proof follows similar steps as Lemma \ref{lem:meta_lemma}.
\end{proof}
After the computation of the approximated moments under \ac{ET} traffic $\tilde{M}_{b,E},;\; b=\{1,2\}$, the meta distribution $F_E(\theta, \xi)$ is evaluated based on (\ref{eq:metaPDF}) after plugging the computed $\tilde{M}_{b,E}$. In addition, the network categorization procedure is carried out in a similar way as explained in Section \ref{QosClass}.
\subsubsection{\textbf{$\mathbf{\varTheta_{E}}$ Characterization for ET Traffic}}
The spatially averaged interfering intensity for the \ac{ET} traffic is equivalent to the percentage of devices which have packets to be transmitted in their respective queues at steady state. To this end, the idle probability of the $n$-th QoS class $x_{0,n}$ captures such activity. Resorting to the mean field theory, $\varTheta_{E}$ is computed by averaging over the $N$ classes temporal idle probabilities  as 
\begin{equation}\label{eq:res}
\varTheta_{E} = \frac{1}{N}\sum_{n=1}^{N} x_{0,n}.
\end{equation}
It is clear that to evaluate $F_E(\theta, \xi)$, one needs first to compute $x_{0,n}$. Such inter-dependency between the network-wide aggregate interference and the queues characterization highlights the cross-relation between the microscopic and macroscopic scales in the network. The characterization of $x_{0,n}$, which is required to evaluate the waiting times and the \ac{PAoI} will be discussed in the following section.

	\thispagestyle{empty}
\section{microscopic queueing theory analysis}\label{sec:QT_anaylsis}
The mathematical model for the microscopic scale (i.e., queue evolution) will be presented in this section. As mentioned, the device's location-dependency is captured via its departure probability (i.e. QoS class dependent), which remains unchanged over long time horizon. In this work, a geometric departure process is adopted to model the packets departure from each device. It is important to note that the geometric departure is an approximation that capitalizes on the negligible temporal correlation of the departure probabilities once the location-dependent \ac{QoS} class is determined as mentioned in Approximation \ref{approxIID}.

\subsection{TT Traffic}
We utilize a degenerate PH type distribution to mimic the \ac{TT} traffic generation at every device.  In particular, the utilized PH type distribution works as a deterministic counter that generates a packet every $T$ time slots. A pictorial illustration of the DTMC with deterministic arrival of packets every $T=4$ time slots is shown in Fig. \ref{fig:DTMC}(a). The PH type distribution is defined as an absorbing Markov chain \cite{Alfa2015}. In the context of \ac{TT}, absorption denotes packet arrival. Mathematically, an absorbing Markov chain is defined as 
\begin{equation}
\mathbf{Q}=\left[ \begin{array}{ll}{1} & \mathbf{0} \\ \mathbf{s} & \mathbf{S} \end{array}\right],
\end{equation}
where $\mathbf{s}\in \mathbb{R}^{T \times 1}$ represents the absorption probability from a given transient phase and is given by $\mathbf{s} = \mathbf{1}_{T}-\mathbf{S}\mathbf{1}_{T}$. The utilized PH type distribution is represented by the tuple $(\bm{\zeta}, \mathbf{S})$, where  $\bm{\zeta} \in \mathbb{R}^{1 \times T}$ is the initialization vector and $\mathbf{S} \in \mathbb{R}^{T \times T}$ is the sub-stochastic transient matrix. In addition, the matrix $\mathbf{S}$ is constructed to count exactly $T$ time slots between two successive packet generations. Accordingly, there is no randomness in the packet generation process and the transition probabilities between the states equal 1. In order to mimic the periodic generation of a packet, $\mathbf{S}$ is formulated as 
\begin{equation}
\mathbf{S}=\left[ 
\begin{array}{lllll}
	0 & 1  & 0 & \cdots & 0  \\
	0 & 0  & 1 & \cdots & 0  \\
	\vdots & \vdots  & \cdots & 0 & 1\\
	0 & 0 & 0 & 0 & 0 
\end{array}\right].
\end{equation}
In addition, $\bm{\zeta} = [1 \; \bm{0}_{T-1}]$. Based on the proposed PH type distribution for the \ac{TT} arrival process, we model the temporal interactions via an PH/Geo/1 queue \cite{Alfa2015}. The departure process is captured via a geometric process due to the adoption of Approximation \ref{approxIID}. Fig. \ref{fig:DTMC}(a) shows the proposed \ac{DTMC} model for the \ac{TT} traffic, where the vertical and horizontal transitions depict transitions between levels and phases, respectively. Utilizing the previously mentioned PH type structure, one can provide a tractable model that captures the queueing temporal dynamics in the form of a \ac{QBD} process~\cite{G.Kulkarni1999}. The queue transitions for a device within the $n$-th class are captured through the \ac{QBD} characterized via the following probability transition matrix  
\begin{equation}\label{eq:QBD}
\mathbf{P}_n=\left[\begin{array}{lllll}{\mathbf{B}} & {\mathbf{C}} & {} & {} & {} \\ {\mathbf{E}_n} & {\mathbf{A}_{1,n}} & {\mathbf{A}_{0,n}} & {} & {} \\ {} & {\mathbf{A}_{2,n}} & {\mathbf{A}_{1.n}} & {\mathbf{A}_{0,n}} & {} \\ {} & {} & {\ddots} & {\ddots} & {\ddots}\end{array}\right],
\end{equation}
where $\mathbf{B}=\mathbf{S}, \mathbf{C}=\mathbf{s}\bm{\zeta}$ and $\mathbf{E}_n = d_n\mathbf{S} \in \mathbb{R}^{T \times T}$ are the boundary sub-stochastic matrices. In addition, $\mathbf{A}_{2,n} = d_n\mathbf{S}, \mathbf{A}_{0,n} = \bar{d}_n\mathbf{s}\bm{\zeta}$, and $\mathbf{A}_{1,n} = d_n\mathbf{s}\bm{\zeta} + \bar{d}_n\mathbf{S} \in \mathbb{R}^{T \times T}$ represent the sub-stochastic matrices that capture the transition down a level, up a level, and in a fixed level within the \ac{QBD}, respectively. In addition, $\mathbf{A}_{2,n}, \mathbf{A}_{1,n}$, and $\mathbf{A}_{0,n}$ are represented via the green, violet, and red arrows in Fig. \ref{fig:DTMC}(a). As mentioned in the previous section, for the \ac{DTMC} in (\ref{eq:QBD}) to be stable, the following condition must be satisfied \cite{Loynes1962}
\begin{equation}\label{eq:stability}
d_n \geq \frac{1}{T-1}.
\end{equation}
In addition, let $\mathbf{x} = [\mathbf{x}_{0} \; \mathbf{x}_{1} \; \mathbf{x}_{2} \; \cdots]$ be the steady state probability vector, where $\mathbf{x}_i = [\mathbf{x}_{i,1} \; \mathbf{x}_{i,2} \; \cdots \; \mathbf{x}_{i,N}]$ and  $\mathbf{x}_{i,j} = [x_{i,j,1} \; x_{i,j,2} \; \cdots x_{i,j,T}]$, where $x_{i,j,k}$ is the probability that a device has $i$ packets and belongs to the $j$-th class and is in the $k$-th arrival phase. In this context, the idle probability of device in the $j$-th class is evaluated as
\begin{equation}\label{eq:x_0}
x_{0,j} = \sum_{n=1}^{T} x_{0,j,n}.
\end{equation}
\begin{figure}
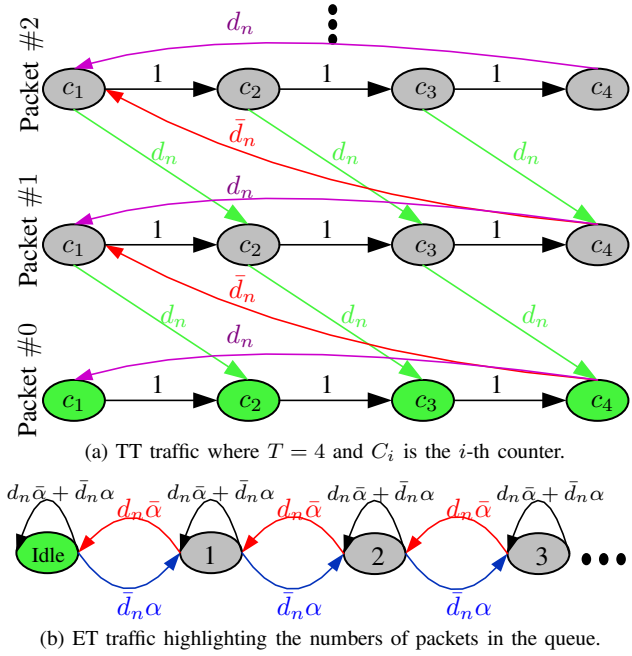

	\ifCLASSOPTIONdraftcls
	\subfloat[TT traffic where $T=4$ and $C_i$ is the $i$-th counter.]{%
		\input{QT_anaylsis/figures/DTMC_diagram/1col/DTMC_diagram_P.tex}
	}
	\subfloat[ET traffic highlighting  the numbers of packets in the queue.]{%
		\input{QT_anaylsis/figures/DTMC_diagram/1col/DTMC_diagram_B.tex}}		
	\else
	\subfloat[TT traffic where $T=4$ and $C_i$ is the $i$-th counter.]{%
		\input{QT_anaylsis/figures/DTMC_diagram/2col/DTMC_diagram_P.tex}
	}\qquad
	\subfloat[ET traffic highlighting the numbers of packets in the queue.]{%
		\input{QT_anaylsis/figures/DTMC_diagram/2col/DTMC_diagram_B.tex}}		
	\fi	
	\caption{DTMCs modeling the temporal evolution. Green states represent idle states.}
	\label{fig:DTMC}
\end{figure}

Through this work, a mathematically tractable solution is sought to address the aforementioned \ac{DTMC} employed at each device. Markov chains with QBD structure can be solved via utilizing the \ac{MAM} \cite{G.Kulkarni1999},\cite{Alfa2015}. Based on the state transition matrix defined in (\ref{eq:QBD}), the following lemma derives the steady state distribution of the queues temporal evolution. 
\begin{lemma}\label{lem:MAM_lemma}
	The steady state distribution of a device belonging to the $n$-th QoS class based on the state transition matrix $\mathbf{P}_n$ under TT traffic with cycle duty $T$ is given by 
	\begin{equation}\label{eq:MAM_SSP}
	\mathbf{x}_{i,n}= \begin{cases}
	\mathbf{x}_{i,n}\mathbf{B} + \mathbf{x}_{i+1,n}\mathbf{E}_n, & i = 0,  \\
	\mathbf{x}_{i-1,n}\mathbf{C} + \mathbf{x}_{i,n}(\mathbf{A}_{1,n} + \mathbf{R}_n\mathbf{A}_{2,n}),  & i =1, \\
	\mathbf{x}_{i-1,n}\mathbf{R}_n, &  i > 1,
	\end{cases}
	\end{equation}
	where $\mathbf{R}_n$ is the \ac{MAM} matrix and is given by $\mathbf{R}_n= \mathbf{A}_{0.n}(\mathbf{I}_{T}-\mathbf{A}_{1,n}-\omega\mathbf{A}_{02,n})^{-1}$. The term $\omega$ is the spectral radius of $\mathbf{R}$, which can be evaluated by solving for $z$ in $z=\mathbf{s}(\mathbf{I}_{T} - \mathbf{A}_{1,n}-z\mathbf{A}_{2,n})^{-1}\mathbf{I}_T$. In addition, (\ref{eq:MAM_SSP}) must  satisfy the  normalization $\mathbf{x}_{0,n}\bm{1}_T + \mathbf{x}_{1,n}(\mathbf{I}_T - \mathbf{R}_n)^{-1}]\bm{1}_T = 1$.
\end{lemma}
\begin{proof}
See Appendix \ref{se:Appendix_B}
\end{proof}
Once the queue distribution is characterized, one can proceed with evaluating the waiting time distribution of a generic packet residing in a queue, which is the major component in computing the \ac{PAoI} as explained in Section \ref{sec:AoI}. Let $\mathcal{W}_n^T$ be the waiting time of a generic packet at a device belonging to the $n$-th \ac{QoS} class in the queue under the \ac{TT} traffic and $\mathcal{W}_n^{m,T} = \mathbb{P}\{\mathcal{W}_n^T = m\}$. Also, let $\mathbf{q}^n_i = [q^n_{i,1} \; q^n_{i,2} \; \cdots \; q^n_{i,T}]$, where $q^n_{i,j}$ is the probability that an incoming packet at a device belonging to the $n$-th class will find $i$ packets waiting and the next packet arrival has phase $j$. In accordance, $\mathbf{q}$ is evaluated as \cite{Alfa2015}
\begin{equation}
\mathbf{q}^n_l= \begin{cases}
\sigma\Big(\mathbf{x}_{i,n}\mathbf{s}\bm{\zeta} + \mathbf{x}_{i+1,n}\mathbf{s}\bm{\zeta}d_n\Big),  &  l = 0 \\ 
\sigma\Big(\mathbf{x}_{i,n}\mathbf{s}\bm{\zeta}\bar{d}_n + \mathbf{x}_{i+1,n}\mathbf{s}\bm{\zeta}d_n\Big), &  l \ge 1,
\end{cases}
\end{equation}
where $\sigma = \bm{\zeta}(\mathbf{I}_T-\mathbf{S})^{-1} \mathbf{1}$. To this end, the waiting time distribution is calculated as 
\begin{equation}\label{eq:waiting_P}
\mathcal{W}_{n}^{m,T}= \begin{cases}
\mathbf{q}_0^n\bm{1}_T,  &  m = 0, \\ 
\sum_{v=1}^{i} \mathbf{q}_{v}^n \bm{1}_T\Big(\begin{array}{c}{i-1} \\ {v-1}\end{array}\Big) b^{v}(1-b)^{i-v}, &  m \ge 1.
\end{cases}
\end{equation}
After computing the waiting time distribution for the TT traffic model, which entails the macroscopic network scale, one can proceed with the \ac{PAoI} evaluation via the following theorem. 
\begin{theorem}\label{lem:PAoI_TT}
	The spatially averaged PAoI under TT traffic with duty cycle $T$ is given by 
	\begin{equation}
	\mathbb{E}\{\Delta_p|  \mathrm{\Phi}, \mathrm{\Psi}\} = T + \frac{1}{N}\Big(\sum_{\varrho=1}^{N}\sum_{j=0}^{\infty} j\mathcal{W}_{\varrho}^{j,T} \Big).
	\end{equation}
\end{theorem}
\begin{proof}
The theorem is proven by plugging  (\ref{eq:waiting_P}) into (\ref{eq:peakAoI}) and noting that $\mathbb{E}^{!}\Big\{\mathcal{I}_o |  \mathrm{\Phi}, \mathrm{\Psi} \Big\} = T$. 
\end{proof}
\subsection{ET Traffic}

The queue evolution for a device within the $n$-th class is captured through a \ac{DTMC} represented in Fig. \ref{fig:DTMC}(b), and characterized via the probability transition matrix $\mathbf{P}_n$ as 
\begin{equation}\label{eq:QBD_ET}
\mathbf{P}_n=\left[\begin{array}{lllll}{\bar{\alpha}} & {\alpha} & {} & {} & {} \\ {\bar{\alpha}d_n} & {\alpha d_n + \bar{\alpha}\bar{d}_n} & {\alpha \bar{d}_n} & {} & {} \\ {} & {\bar{\alpha}d_n} & {\alpha d_n + \bar{\alpha}\bar{d}_n} & {\alpha \bar{d}_n} & {} \\ {} & {} & {\ddots} & {\ddots} & {\ddots}\end{array}\right].
\end{equation}
For the \ac{ET} traffic, the \ac{DTMC} in (\ref{eq:QBD_ET}) is stable if the inequality $\frac{\alpha}{d_n} < 1$ is satisfied. For unstable \acp{DTMC}, the idle probability is naturally 0. To this end, let $\mathbf{x}_n = [x_{0,n} \; x_{1,n} \; x_{2,n} \; \cdots]$ be the steady state probability vector of the $n$-th class, where $x_{i,n}$ is the probability that a device belonging to the $n$-th class has $i$ packets residing in its queue. The idle probability of device in the $j$-th class is evaluated as \cite{Alfa2015}
\begin{equation}\label{eq:x_0}
x_{i,n} = R_n^i\frac{x_{0,n}}{\bar{d}_n}, \;\; \text{where } R_n = \frac{\alpha \bar{d}_n}{\bar{\alpha}d_n}, \;\; \text{and } x_{0,n} = \frac{d_n - \alpha}{d_n}.
\end{equation}
Once the queue distribution is characterized, one can proceed with evaluating the spatially averaged idle probability $\varTheta_{E}$ and the waiting time distribution of a generic packet within the considered queue. Similar to the \ac{TT}  traffic, an inter-dependency exists between the network-wide aggregate interference (i.e., $F_E(\theta, \alpha)$) and the queues characterization (i.e., $\varTheta_{E}$). To solve such interdependency, Algorithm 1 is presented which provides a uniquely determined solution by virtue of fixed point theorem. 

As mentioned earlier, let $\mathcal{W}_n^E$ be the waiting time of a generic packet at a device belonging to the $n$-th \ac{QoS} class in the queue under \ac{ET} traffic and $\mathcal{W}_n^{m,E} = \mathbb{P}\{\mathcal{W}_n^E = m\}$. The waiting time for the $n$-th class is \cite{Alfa2015}
\begin{equation}\label{eq:waiting_B}
\mathcal{W}_{n}^{m,E}= \begin{cases}
\frac{d_n -\alpha}{d_n},  &  m = 0, \\ 
\sum_{v=1}^{i} x_{v,n}\left(\begin{array}{c}{i-1} \\ {v-1}\end{array}\right) d_n^{v}(1-d_n)^{i-v}, &  m \ge 1.
\end{cases}
\end{equation}

\begin{figure}
	\begin{algorithm}[H]\label{alg:iterative}
		\caption{Computation of ${F}_E(\theta, \delta)$}\label{euclid}
		\begin{algorithmic}
			\Procedure{}{$\alpha, \epsilon, \theta, N, \varphi$} 
			\LState initialize  $\varTheta_{E}$ 
			\While {$||\varTheta_{E}^k - \varTheta_{E}^{k-1}|| \geq \varphi$}  
			\LState Compute the moments $\tilde{M}_{b,E}$ from Lemma \ref{lem:meta_lemma_2}				 
			\LState Evaluate $F_E(\theta, \xi)$ based on (\ref{eq:metaPDF})
			\LState Compute $d_i,\forall i=\{1,2,\cdots,N\}$ from the Discretized 
			\LState $F_E(\theta, \xi)$ based on (\ref{eq:pmf1}) and (\ref{eq:pmf2}) 
			\For {$n=\{1,2,\cdots,N\}$} 
			\If {$\alpha < d_n$} \Comment{Stability condition}
			\LState Compute $x_{0,n}$ based on (\ref{eq:x_0})
			\Else
			\LState Set $x_{0,n} = 0$
			\EndIf
			\EndFor
			\LState Compute $\varTheta_{E}$ based on (\ref{eq:res})
			\LState Increment k
			\EndWhile
			\LState \Return  ${F}_E(\theta, \delta)$
			\EndProcedure
		\end{algorithmic}
	\end{algorithm}
\end{figure}
Finally, the following theorem characterizes the \ac{PAoI} under \ac{ET} traffic.
\begin{theorem}\label{lem:PAoI_ET}
	The spatially averaged PAoI under ET traffic with cycle duty $T$ is given by 
	\begin{equation}
	\mathbb{E}\{\Delta_p|  \mathrm{\Phi}, \mathrm{\Psi}\} = \frac{1}{\alpha} + \frac{1}{N}\Big(\sum_{\varrho=1}^{N}\sum_{j=0}^{\infty} j\mathcal{W}_{\varrho}^{j,E} \Big).
	\end{equation}
\end{theorem}
\begin{proof}
	The theorem is proven by plugging  (\ref{eq:waiting_B}) into (\ref{eq:peakAoI}) and noting that $\mathbb{E}^{!}\Big\{\mathcal{I}_o |  \mathrm{\Phi}, \mathrm{\Psi} \Big\} = \frac{1}{\alpha}$.
\end{proof}

	\thispagestyle{empty}
\section{Numerical Results}\label{sec:simulation_results}

In this section, different numerical insights are presented for the purpose of (a) validating the proposed mathematical framework for the two traffic models, (b) characterizing the information freshness within a large scale uplink \ac{IoT} network, and (c) highlighting the influence of the system parameters on the network's stability. First, discussion of the simulation environment is presented to establish a clear understanding of the simulation framework. 

\subsection{Simulation Methodology}

The established simulation framework involves deployment of \acp{BS} and devices as discussed in Section \ref{sec:system_model}. Ergodicity is ensured via microscopic averaging, in which the temporal steady state statistics of the queues at each device are collected. The simulation area is $10 \times 10 \text{ km}^2$ with a wrapped-around boundaries to eliminate the effect of the boundary devices within the network. Discretized, synchronized, and time-slotted system is considered, where during each time slot (i.e., microscopic run), independent channel gains are instantiated and packets are generated deterministically or probabilistically, depending on the traffic model. At the start of the simulation, for the \ac{TT} traffic, all the devices within the network are assigned an i.i.d. transmission offset $\beta_i$ from the distribution $f_\beta(\tau)=\frac{1}{T}$ for $\tau \in \{0,1,\cdots,T-1\}$, which depicts the time index of a packet generation event. A new packet is generated periodically following $\beta_o + \ell T, \; \forall \ell = 1,2,\cdots$. For the \ac{ET} traffic, a new packet is generated at each device every time slot with the probability $\alpha$. Every device with packets residing in its queue attempts the communication of such packets with its serving \ac{BS} based on a \ac{FCFS} strategy. A packet is dropped from its queue if the realized uplink \ac{SIR} at the serving \ac{BS} is greater than the detection threshold $\theta$.
\begin{figure}
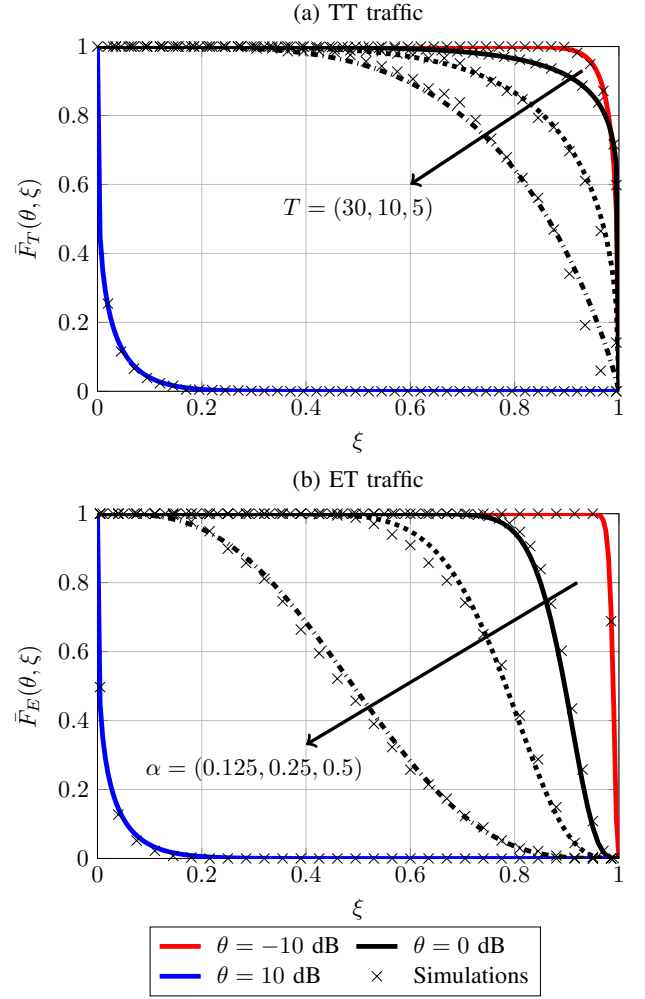

	\centering
	\ifCLASSOPTIONdraftcls
	\input{simulation_results/figures/meta_sim/1col/meta_sim_com.tex}	
	\else
	\input{simulation_results/figures/meta_sim/2col/meta_sim_com.tex}	
	\fi
	\caption{Meta distribution for TT and ET traffic models.}
	\label{fig:meta_sim} 
\end{figure}
To ensure a steady state operation of the queues, each queue's occupancy at each device is monitored. For initialization, all queues at the devices are initiated as being empty and then simulation runs for a sufficiently large number of time slots till steady-state is realized. Let $\hat{x}_0^t$ denotes the average idle steady state probability across all the devices within the network for the $t$-th iteration. Mathematically, the steady state behavior is reached once $||\hat{x}_0^k-\hat{x}_0^{k-1}|| < \varphi$, where $\varphi$ is some predefined tolerance (e.g. $10^{-4}$). Once steady state is reached, all temporal statistics are then collected based on adequately large number of microscopic realizations (e.g., $10000$). Unless otherwise stated, we consider the following parameters: $\eta = 4$, $\rho = -90$ dBm, $\epsilon = 1$, $T=8$ and $\alpha = 0.125$. 

In Fig. \ref{fig:meta_sim}, we consider the framework verification via the meta distribution of the transmission success probability for each traffic model with different traffic loads and detection thresholds. First, for the two considered traffic models, one can observe a close match between the simulation and the proposed analytical framework, which confirms the accuracy of the proposed mathematical model and shows that the interdependency between the network-wide aggregate interference and the queues temporal evolution is captured. For low values of $\theta$, the devices are able to empty their queues and become idle. This leads to a lower network-wide aggregate interference, and thus increased percentile of devices achieving a given reliability $\xi$. As $\theta$ increases, the probability of successful transmission attempts for a generic device decreases, which aggravates the aggregate network interference. Consequently, more devices are active within the network and the achieved reliability to meet the targeted $\theta$ decreases. Fig. \ref{fig:meta_sim}(a) presents the \ac{TT} traffic patters for different values of cyclic duration. It is observed that as $T$ decreases, the percentile of active devices increases within the network. Decreasing $T$ increases the packet generation rate, shortens the time required to dispatch generated packets, and increases the number of synchronized devices. Accordingly, the network interference increases, which deteriorates the transmission success probabilities. Such a consequential effect of increased traffic load affects the percentile of devices within the network to achieve a given transmission success probability, as illustrated via the meta distribution. In addition, Fig. \ref{fig:meta_sim}(b) presents the meta distribution for the \ac{ET} traffic model with different arrival probabilities. Similar to \ac{TT} case, as $\alpha$ increases, the percentile of active devices increases within the network, thus affecting the reliability to achieve a targeted decoding threshold $\theta$. More insights comparing the \ac{TT} to the \ac{ET} models will be discussed in Fig. \ref{fig:det_geo_comp}.

\begin{figure} 
	\centering
	\ifCLASSOPTIONdraftcls
	\definecolor{mycolor1}{rgb}{0.0, 0.0, 1.0}
\definecolor{mycolor2}{rgb}{1.0, 0.01, 0.24}
\begin{tikzpicture}
\pgfplotsset{
	width=0.75\columnwidth,
	height=1.8in,
	scale only axis,
	xmin=3,
	xmax=20,
	legend style={%
		legend columns =2,
		at={(0.7,0.99)},
		anchor=north east
	},
	legend cell align=left,
	ticklabel shift={0.05cm},
	tick label style={/pgf/number format/1000 sep=}
}
\begin{axis}[%
every outer y axis line/.append style={mycolor2},
every y tick label/.append style={font=\color{mycolor2}},
every y tick/.append style={mycolor2},
ymin=2,
ymax=20,
ylabel style={font=\color{mycolor2}},
ylabel={Waiting time},
xlabel={$T$ [Time slots]},
yticklabel pos=right
]
\addlegendimage{style={color=black,dotted, line width=2.0pt}}
\addlegendentry{{ $\theta = -5$ dB}}
\addlegendimage{style={color=black,line width=2.0pt}}
\addlegendentry{{ $\theta = 0$ dB}}
\addlegendimage{style={color=black,
		mark=*, mark options={solid} , line width=2.0pt]}}
\addlegendentry{{ $\theta = 5$ dB}}
\addlegendimage{color=black, only marks, mark size=3pt, mark=*, mark options={solid, fill=green}}
\addlegendentry{{Stability point}}

\addplot [color=mycolor2, mark=*, mark options={solid}, line width=2.0pt]
  table[row sep=crcr]{%
4	inf\\
5	inf\\
6	inf\\
7	inf\\
8	9.17434048693152\\
9	6.36737835939597\\
10	5.4673291105167\\
11	4.4365041494544\\
12	4.08934080651838\\
13	3.68418165499784\\
14	3.53356664147004\\
15	3.3378138007776\\
16	3.36828385364703\\
17	2.99316988074475\\
18	3.00650681365453\\
19	3.0187664918623\\
20	3.03006035547106\\
};

\addplot [color=black, only marks, mark size=3pt, mark=*, mark options={solid, fill=green}]
table[row sep=crcr]{%
	8	9.17434048693152\\
};

\addplot [color=mycolor2, line width=2.0pt]
  table[row sep=crcr]{%
4	19.1376557304941\\
5	9.23544385041168\\
6	6.35558458062669\\
7	5.2495330627464\\
8	4.14954580520336\\
9	3.91369310628236\\
10	3.40975552736508\\
11	3.42259932390413\\
12	3.17992057419832\\
13	3.18633472987705\\
14	3.19215477494114\\
15	3.19745226951302\\
16	2.97706585333942\\
17	2.97857074685161\\
18	2.97995183335213\\
19	2.98122143356389\\
20	2.98238996391325\\
};
\addplot [color=black, only marks, mark size=3pt, mark=*, mark options={solid, fill=green}]
table[row sep=crcr]{%
	4	19.1376557304941\\
};

\addplot [color=mycolor2, dotted, line width=2.0pt]
  table[row sep=crcr]{%
4	7.38980586639601\\
5	5.7183815004394\\
6	4.79174277939332\\
7	3.88406284061822\\
8	3.59779260768809\\
9	3.31674068315681\\
10	3.31333793036225\\
11	3.04713111466424\\
12	3.04315358256758\\
13	3.03984496428581\\
14	3.03704734721854\\
15	3.03465264171304\\
16	3.03257876214666\\
17	3.03076490420858\\
18	3.02916521339915\\
19	3.02774375069987\\
20	3.02647143123135\\
};
\addplot [color=black, only marks, mark size=3pt, mark=*, mark options={solid, fill=green}]
table[row sep=crcr]{%
4	7.38980586639601\\
};

\end{axis}

\begin{axis}[%
every y tick label/.append style={font=\color{mycolor1}},
every y tick/.append style={mycolor1},
ymin=10,
ymax=25,
ylabel style={font=\color{mycolor1}},
ylabel={PAoI},
yticklabel pos=left,
xtick=\empty
]
\addplot [color=mycolor1, line width=2.0pt]
table[row sep=crcr]{%
	4	23.1376557304941\\
	5	14.2354438504117\\
	6	12.3555845806267\\
	7	12.2495330627464\\
	8	12.1495458052034\\
	9	12.9136931062824\\
	10	13.4097555273651\\
	11	14.4225993239041\\
	12	15.1799205741983\\
	13	16.1863347298771\\
	14	17.1921547749411\\
	15	18.197452269513\\
	16	18.9770658533394\\
	17	19.9785707468516\\
	18	20.9799518333521\\
	19	21.9812214335639\\
	20	22.9823899639133\\
};

\addplot [color=black, only marks, mark size=3pt, mark=*, mark options={solid, fill=green}]
table[row sep=crcr]{%
	4	23.1376557304941\\
};
\addplot [color=mycolor1, mark=*, mark options={solid}, line width=2.0pt]
table[row sep=crcr]{%
	4	inf\\
	5	inf\\
	6	inf\\
	7	inf\\
	8	17.1743404869315\\
	9	15.367378359396\\
	10	15.4673291105167\\
	11	15.4365041494544\\
	12	16.0893408065184\\
	13	16.6841816549978\\
	14	17.53356664147\\
	15	18.3378138007776\\
	16	19.368283853647\\
	17	19.9931698807448\\
	18	21.0065068136545\\
	19	22.0187664918623\\
	20	23.0300603554711\\
};

\addplot [color=black, only marks, mark size=3pt, mark=*, mark options={solid, fill=green}]
table[row sep=crcr]{%
	8	17.1743404869315\\
};

\addplot [color=mycolor1, dotted, line width=2.0pt]
table[row sep=crcr]{%
	4	11.389805866396\\
	5	10.7183815004394\\
	6	10.7917427793933\\
	7	10.8840628406182\\
	8	11.5977926076881\\
	9	12.3167406831568\\
	10	13.3133379303623\\
	11	14.0471311146642\\
	12	15.0431535825676\\
	13	16.0398449642858\\
	14	17.0370473472185\\
	15	18.034652641713\\
	16	19.0325787621467\\
	17	20.0307649042086\\
	18	21.0291652133991\\
	19	22.0277437506999\\
	20	23.0264714312313\\
};

\addplot [color=black, only marks, mark size=3pt, mark=*, mark options={solid, fill=green}]
table[row sep=crcr]{%
	4	11.389805866396\\
};

\draw [black] (rel axis cs:0.7,0.55) ellipse [x radius=0.5, y radius=1];
\node at (rel axis cs:0.65,0.65) {\textcolor{blue}{PAoI}};

\draw [black] (rel axis cs:0.7,0.08) ellipse [x radius=0.5, y radius=1];
\node at (rel axis cs:0.7,0.2) {\textcolor{red}{Waiting time}};

\end{axis}
\end{tikzpicture}%
	\else
	\definecolor{mycolor1}{rgb}{0.0, 0.0, 1.0}
\definecolor{mycolor2}{rgb}{1.0, 0.01, 0.24}
\begin{tikzpicture}
\pgfplotsset{
	width=0.75\columnwidth,
	height=2in,
	scale only axis,
	xmin=3,
	xmax=20,
	legend style={%
		legend columns =2,
		at={(0.9,0.99)},
		anchor=north east
	},
	legend cell align=left,
	ticklabel shift={0.05cm},
	tick label style={/pgf/number format/1000 sep=}
}
\begin{axis}[%
every outer y axis line/.append style={mycolor2},
every y tick label/.append style={font=\color{mycolor2}},
every y tick/.append style={mycolor2},
ymin=2,
ymax=20,
ylabel style={font=\color{mycolor2}},
ylabel={Waiting time},
xlabel={$T$ [Time slots]},
yticklabel pos=right
]
\addlegendimage{style={color=black,dotted, line width=2.0pt}}
\addlegendentry{\footnotesize{ $\theta = -5$ dB}}
\addlegendimage{style={color=black,line width=2.0pt}}
\addlegendentry{\footnotesize{ $\theta = 0$ dB}}
\addlegendimage{style={color=black,
		mark=*, mark options={solid} , line width=2.0pt]}}
\addlegendentry{\footnotesize{ $\theta = 5$ dB}}
\addlegendimage{color=black, only marks, mark size=3pt, mark=*, mark options={solid, fill=green}}
\addlegendentry{\footnotesize{Stability point}}

\addplot [color=mycolor2, mark=*, mark options={solid}, line width=2.0pt]
  table[row sep=crcr]{%
4	inf\\
5	inf\\
6	inf\\
7	inf\\
8	9.17434048693152\\
9	6.36737835939597\\
10	5.4673291105167\\
11	4.4365041494544\\
12	4.08934080651838\\
13	3.68418165499784\\
14	3.53356664147004\\
15	3.3378138007776\\
16	3.36828385364703\\
17	2.99316988074475\\
18	3.00650681365453\\
19	3.0187664918623\\
20	3.03006035547106\\
};

\addplot [color=black, only marks, mark size=3pt, mark=*, mark options={solid, fill=green}]
table[row sep=crcr]{%
	8	9.17434048693152\\
};

\addplot [color=mycolor2, line width=2.0pt]
  table[row sep=crcr]{%
4	19.1376557304941\\
5	9.23544385041168\\
6	6.35558458062669\\
7	5.2495330627464\\
8	4.14954580520336\\
9	3.91369310628236\\
10	3.40975552736508\\
11	3.42259932390413\\
12	3.17992057419832\\
13	3.18633472987705\\
14	3.19215477494114\\
15	3.19745226951302\\
16	2.97706585333942\\
17	2.97857074685161\\
18	2.97995183335213\\
19	2.98122143356389\\
20	2.98238996391325\\
};
\addplot [color=black, only marks, mark size=3pt, mark=*, mark options={solid, fill=green}]
table[row sep=crcr]{%
	4	19.1376557304941\\
};

\addplot [color=mycolor2, dotted, line width=2.0pt]
  table[row sep=crcr]{%
4	7.38980586639601\\
5	5.7183815004394\\
6	4.79174277939332\\
7	3.88406284061822\\
8	3.59779260768809\\
9	3.31674068315681\\
10	3.31333793036225\\
11	3.04713111466424\\
12	3.04315358256758\\
13	3.03984496428581\\
14	3.03704734721854\\
15	3.03465264171304\\
16	3.03257876214666\\
17	3.03076490420858\\
18	3.02916521339915\\
19	3.02774375069987\\
20	3.02647143123135\\
};
\addplot [color=black, only marks, mark size=3pt, mark=*, mark options={solid, fill=green}]
table[row sep=crcr]{%
4	7.38980586639601\\
};

\end{axis}

\begin{axis}[%
every y tick label/.append style={font=\color{mycolor1}},
every y tick/.append style={mycolor1},
ymin=10,
ymax=25,
ylabel style={font=\color{mycolor1}},
ylabel={PAoI},
yticklabel pos=left,
xtick=\empty
]
\addplot [color=mycolor1, line width=2.0pt]
table[row sep=crcr]{%
	4	23.1376557304941\\
	5	14.2354438504117\\
	6	12.3555845806267\\
	7	12.2495330627464\\
	8	12.1495458052034\\
	9	12.9136931062824\\
	10	13.4097555273651\\
	11	14.4225993239041\\
	12	15.1799205741983\\
	13	16.1863347298771\\
	14	17.1921547749411\\
	15	18.197452269513\\
	16	18.9770658533394\\
	17	19.9785707468516\\
	18	20.9799518333521\\
	19	21.9812214335639\\
	20	22.9823899639133\\
};

\addplot [color=black, only marks, mark size=3pt, mark=*, mark options={solid, fill=green}]
table[row sep=crcr]{%
	4	23.1376557304941\\
};
\addplot [color=mycolor1, mark=*, mark options={solid}, line width=2.0pt]
table[row sep=crcr]{%
	4	inf\\
	5	inf\\
	6	inf\\
	7	inf\\
	8	17.1743404869315\\
	9	15.367378359396\\
	10	15.4673291105167\\
	11	15.4365041494544\\
	12	16.0893408065184\\
	13	16.6841816549978\\
	14	17.53356664147\\
	15	18.3378138007776\\
	16	19.368283853647\\
	17	19.9931698807448\\
	18	21.0065068136545\\
	19	22.0187664918623\\
	20	23.0300603554711\\
};

\addplot [color=black, only marks, mark size=3pt, mark=*, mark options={solid, fill=green}]
table[row sep=crcr]{%
	8	17.1743404869315\\
};

\addplot [color=mycolor1, dotted, line width=2.0pt]
table[row sep=crcr]{%
	4	11.389805866396\\
	5	10.7183815004394\\
	6	10.7917427793933\\
	7	10.8840628406182\\
	8	11.5977926076881\\
	9	12.3167406831568\\
	10	13.3133379303623\\
	11	14.0471311146642\\
	12	15.0431535825676\\
	13	16.0398449642858\\
	14	17.0370473472185\\
	15	18.034652641713\\
	16	19.0325787621467\\
	17	20.0307649042086\\
	18	21.0291652133991\\
	19	22.0277437506999\\
	20	23.0264714312313\\
};

\addplot [color=black, only marks, mark size=3pt, mark=*, mark options={solid, fill=green}]
table[row sep=crcr]{%
	4	11.389805866396\\
};
\draw [black] (rel axis cs:0.7,0.55) ellipse [x radius=0.5, y radius=1];
\node at (rel axis cs:0.65,0.65) {\textcolor{blue}{PAoI}};

\draw [black] (rel axis cs:0.7,0.08) ellipse [x radius=0.5, y radius=1];
\node at (rel axis cs:0.7,0.2) {\textcolor{red}{Waiting time}};

\end{axis}
\end{tikzpicture}%
	\fi
	\caption{PAoI (left) and average waiting time (right) for TT traffic with increasing duty cycle $T$ and different $\theta$.}
	\label{fig:AoI__WT_P} 
\end{figure}
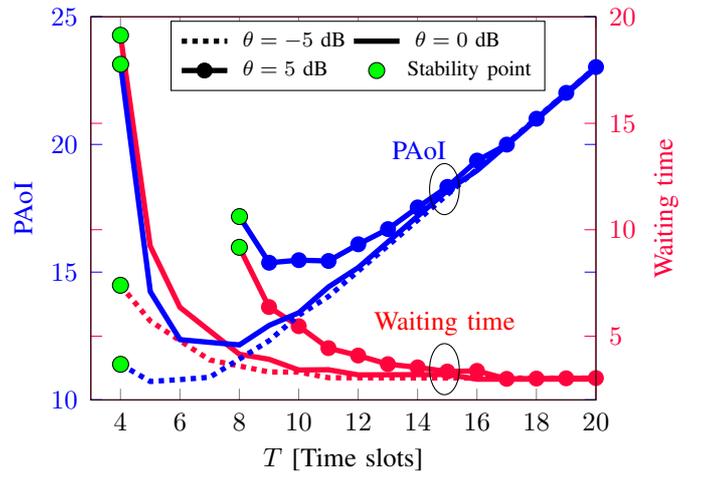 

Fig. \ref{fig:AoI__WT_P} plots the spatially averaged \ac{PAoI} along with average waiting time for versus the cycle duration $T$ for the \ac{TT} traffic model. As explained in Section \ref{sec:AoI}, the  \ac{PAoI} is sensitive to the inter-arrival and system waiting times of a randomly selected packet within the queue. First we investigate the effect of $\theta$. As $\theta$ increases, packets transmission success is subjected to a more stringent requirement on the achieved SIR. This leads to increased retransmissions, thus, increasing the mutual interference due to lower idle probabilities. The increased mutual interference hinders the successful departure of the packets from their respective queues and lead to queue instability in some devices, yielding instability (i.e., infinite waiting times and \ac{PAoI}). The figure also shows the effect of the cycle times. For high values of $T$, the large inter-arrival times is the dominant factor, yielding high values of \ac{PAoI}, while the waiting time is low. Low values of waiting times are the result of having sufficient time to transmit a residing packet, before the event of a new packet arrival. As $T$ decreases, the waiting times dominates, yielding an increase in the \ac{PAoI} till point of queue instability, as indicated by the stability point. Consequently, adopting a \ac{TT} traffic with duty cycle $T<4$ results in an unstable system and infinite \ac{PAoI}. The effect of $\theta$ on the stability frontiers can be explained in a similar fashion to that of Fig. \ref{fig:meta_sim}, where increasing $\theta$ diminishes the stability region due to the increased network-wide aggregate interference. While reduced traffic arrivals reliefs network interference and reduces delay, it is not the case for AoI because it prolongs the updates duty cycle.  Hence, there is an optimal duty cycle that minimizes the PAoI by balancing the tradeoff between frequency of updates and the aggregate network interference.

\begin{figure} 
	\centering
	\ifCLASSOPTIONdraftcls
	\definecolor{mycolor1}{rgb}{0.0, 0.0, 1.0}
\definecolor{mycolor2}{rgb}{1.0, 0.01, 0.24}
\begin{tikzpicture}
\pgfplotsset{
	width=0.75\columnwidth,
	height=1.8in,
	scale only axis,
	xmin=0.05,
	xmax=0.7,
	xticklabel style={
		/pgf/number format/fixed,
		/pgf/number format/precision=4
	},
	legend style={legend columns =2,
		at={(0.7,0.99)},
		anchor=north east
	},
	legend cell align=left,
	ticklabel shift={0.05cm},
	tick label style={/pgf/number format/1000 sep=}
}
\begin{axis}[%
every outer y axis line/.append style={mycolor2},
every y tick label/.append style={font=\color{mycolor2}},
every y tick/.append style={mycolor2},
ymin=1,
ymax=6,
ylabel style={font=\color{mycolor2}},
ylabel={Waiting time},
xlabel={$\alpha$ [packets/slot]},
yticklabel pos=right
]
\addlegendimage{legend columns=2, style={color=black, dotted, line width=2.0pt]}}
\addlegendentry{$\theta = -5 \text{ dB}$}
\addlegendimage{style={color=black, line width=2.0pt}}
\addlegendentry{$\theta = 0 \text{ dB}$}
\addlegendimage{style={color=black,
		mark=*, mark options={solid}, line width=2.0pt}}
\addlegendentry{$\theta = 5 \text{ dB}$}

\addplot [color=black, only marks, mark size=3pt, mark=*, mark options={solid, fill=green}]
table[row sep=crcr]{%
	0.15	2.13601130109271\\
};\addlegendentry{{Stability point}}

\addplot [color=mycolor2, mark=*, mark options={solid},line width=2.0pt]
table[row sep=crcr]{%
	0.05	1.17360520490003\\
	0.07	1.26722974181006\\
	0.09	1.38401587897303\\
	0.11	1.53748731610666\\
	0.13	1.75772224453968\\
	0.15	2.13601130109271\\
	0.17	inf\\
	0.19	inf\\
	0.21	inf\\
	0.23	inf\\
	0.25	inf\\
	0.27	inf\\
	0.29	inf\\
	0.31	inf\\
	0.33	inf\\
	0.35	inf\\
	0.37	inf\\
	0.39	inf\\
	0.41	inf\\
	0.43	inf\\
	0.45	inf\\
	0.47	inf\\
	0.49	inf\\
	0.51	inf\\
	0.53	inf\\
	0.55	inf\\
	0.57	inf\\
	0.59	inf\\
	0.61	inf\\
	0.63	inf\\
	0.65	inf\\
	0.67	inf\\
	0.69	inf\\
	0.71	inf\\
	0.73	inf\\
	0.75	inf\\
	0.77	inf\\
	0.79	inf\\
};
\addplot [color=black, only marks, mark size=3pt, mark=*, mark options={solid, fill=green}]
table[row sep=crcr]{%
	0.35	4.48216917628451\\
};
\addplot [color=black, only marks, mark size=3pt, mark=*, mark options={solid, fill=green}]
table[row sep=crcr]{%
	0.15	2.13601130109271\\
};

\addplot [color=mycolor2, line width=2.0pt]
table[row sep=crcr]{%
	0.05	1.09615375938418\\
	0.07	1.13935830606912\\
	0.09	1.18592187637567\\
	0.11	1.23641776944209\\
	0.13	1.29158608006148\\
	0.15	1.35235526369897\\
	0.17	1.42002093006921\\
	0.19	1.49631010569694\\
	0.21	1.58362167583047\\
	0.23	1.68563005643157\\
	0.25	1.80794772258687\\
	0.27	1.9599426257718\\
	0.29	2.15950128915614\\
	0.31	2.44583173683114\\
	0.33	2.93680803069121\\
	0.35	4.48216917628451\\
	0.37	inf\\
	0.39	inf\\
	0.41	inf\\
	0.43	inf\\
	0.45	inf\\
	0.47	inf\\
	0.49	inf\\
	0.51	inf\\
	0.53	inf\\
	0.55	inf\\
	0.57	inf\\
	0.59	inf\\
	0.61	inf\\
	0.63	inf\\
	0.65	inf\\
	0.67	inf\\
	0.69	inf\\
	0.71	inf\\
	0.73	inf\\
	0.75	inf\\
	0.77	inf\\
	0.79	inf\\
};

\addplot [color=mycolor2, dotted, line width=2.0pt]
table[row sep=crcr]{%
	0.05	1.06612918647197\\
	0.07	1.09372447480384\\
	0.09	1.12204582214216\\
	0.11	1.15115125126206\\
	0.13	1.18110523391927\\
	0.15	1.21198065279054\\
	0.17	1.24386170039379\\
	0.19	1.27684568542407\\
	0.21	1.311041837097\\
	0.23	1.34658104244248\\
	0.25	1.38360930045667\\
	0.27	1.42231570170467\\
	0.29	1.46291350035142\\
	0.31	1.50565963631341\\
	0.33	1.5508728051319\\
	0.35	1.59894494001488\\
	0.37	1.65036779796704\\
	0.39	1.70577751842421\\
	0.41	1.76600675985318\\
	0.43	1.83213814005848\\
	0.45	1.90568622713971\\
	0.47	1.98889854047503\\
	0.49	2.08497552375595\\
	0.51	2.1991781703361\\
	0.53	2.34049783957242\\
	0.55	2.5266598304001\\
	0.57	2.79998995512125\\
	0.59	3.30689162929792\\
	0.61	5.47158987551092\\
	0.63	inf\\
	0.65	inf\\
	0.67	inf\\
	0.69	inf\\
	0.71	inf\\
	0.73	inf\\
	0.75	inf\\
	0.77	inf\\
	0.79	inf\\
};

\addplot [color=black, only marks, mark size=3pt, mark=*, mark options={solid, fill=green}]
table[row sep=crcr]{%
	0.61	5.47158987551092\\
};

\end{axis}

\begin{axis}[%
every y tick label/.append style={font=\color{mycolor1}},
every y tick/.append style={mycolor1},
ymin=4,
ymax=22,
ylabel style={font=\color{mycolor1}},
ylabel={PAoI},
yticklabel pos=left,
xtick=\empty
]
\addplot [color=mycolor1, mark=*, mark options={solid}, line width=2.0pt]
table[row sep=crcr]{%
	0.05	21.1736052049\\
	0.07	15.5529440275243\\
	0.09	12.4951269900841\\
	0.11	10.6283964070157\\
	0.13	9.45002993684737\\
	0.15	8.80267796775938\\
	0.17	inf\\
	0.19	inf\\
	0.21	inf\\
	0.23	inf\\
	0.25	inf\\
	0.27	inf\\
	0.29	inf\\
	0.31	inf\\
	0.33	inf\\
	0.35	inf\\
	0.37	inf\\
	0.39	inf\\
	0.41	inf\\
	0.43	inf\\
	0.45	inf\\
	0.47	inf\\
	0.49	inf\\
	0.51	inf\\
	0.53	inf\\
	0.55	inf\\
	0.57	inf\\
	0.59	inf\\
	0.61	inf\\
	0.63	inf\\
	0.65	inf\\
	0.67	inf\\
	0.69	inf\\
	0.71	inf\\
	0.73	inf\\
	0.75	inf\\
	0.77	inf\\
	0.79	inf\\
};

\addplot [color=black, only marks, mark size=3pt, mark=*, mark options={solid, fill=green}]
table[row sep=crcr]{%
	0.15	8.80267796775938\\
};
\addplot [color=mycolor1, line width=2.0pt]
table[row sep=crcr]{%
	0.05	21.0961537593842\\
	0.07	15.4250725917834\\
	0.09	12.2970329874868\\
	0.11	10.3273268603512\\
	0.13	8.98389377236917\\
	0.15	8.01902193036564\\
	0.17	7.30237387124569\\
	0.19	6.75946800043378\\
	0.21	6.34552643773523\\
	0.23	6.03345614338809\\
	0.25	5.80794772258687\\
	0.27	5.66364632947551\\
	0.29	5.6077771512251\\
	0.31	5.67163818844404\\
	0.33	5.96711106099424\\
	0.35	7.33931203342737\\
	0.37	inf\\
	0.39	inf\\
	0.41	inf\\
	0.43	inf\\
	0.45	inf\\
	0.47	inf\\
	0.49	inf\\
	0.51	inf\\
	0.53	inf\\
	0.55	inf\\
	0.57	inf\\
	0.59	inf\\
	0.61	inf\\
	0.63	inf\\
	0.65	inf\\
	0.67	inf\\
	0.69	inf\\
	0.71	inf\\
	0.73	inf\\
	0.75	inf\\
	0.77	inf\\
	0.79	inf\\
};

\addplot [color=black, only marks, mark size=3pt, mark=*, mark options={solid, fill=green}]
table[row sep=crcr]{%
	0.35	7.33931203342737\\
};

\addplot [color=mycolor1, dotted, line width=2.0pt]
table[row sep=crcr]{%
	0.05	21.066129186472\\
	0.07	15.3794387605181\\
	0.09	12.2331569332533\\
	0.11	10.2420603421712\\
	0.13	8.87341292622696\\
	0.15	7.87864731945721\\
	0.17	7.12621464157026\\
	0.19	6.54000358016092\\
	0.21	6.07294659900176\\
	0.23	5.694407129399\\
	0.25	5.38360930045667\\
	0.27	5.12601940540837\\
	0.29	4.91118936242039\\
	0.31	4.73146608792632\\
	0.33	4.58117583543493\\
	0.35	4.45608779715774\\
	0.37	4.35307050066974\\
	0.39	4.26988008252678\\
	0.41	4.20503115009709\\
	0.43	4.15771953540732\\
	0.45	4.12790844936193\\
	0.47	4.11655811494311\\
	0.49	4.12579185028657\\
	0.51	4.15996248406159\\
	0.53	4.22729029240261\\
	0.55	4.34484164858192\\
	0.57	4.55437592003353\\
	0.59	5.00180688353521\\
	0.61	7.110934137806\\
	0.63	inf\\
	0.65	inf\\
	0.67	inf\\
	0.69	inf\\
	0.71	inf\\
	0.73	inf\\
	0.75	inf\\
	0.77	inf\\
	0.79	inf\\
};

\addplot [color=black, only marks, mark size=3pt, mark=*, mark options={solid, fill=green}]
table[row sep=crcr]{%
	0.61	7.110934137806\\
};
\end{axis}
\end{tikzpicture}%
	\else
	\definecolor{mycolor1}{rgb}{0.0, 0.0, 1.0}
\definecolor{mycolor2}{rgb}{1.0, 0.01, 0.24}
\begin{tikzpicture}
\pgfplotsset{
	width=0.75\columnwidth,
	height=2in,
	scale only axis,
	xmin=0.05,
	xmax=0.7,
	legend style={legend columns =2,
		at={(0.8,0.99)},
		anchor=north east
	},
	legend cell align=left,
	ticklabel shift={0.05cm},
	tick label style={/pgf/number format/1000 sep=}
}
\begin{axis}[%
every outer y axis line/.append style={mycolor2},
every y tick label/.append style={font=\color{mycolor2}},
every y tick/.append style={mycolor2},
ymin=1,
ymax=6,
ylabel style={font=\color{mycolor2}},
ylabel={Waiting time},
xlabel={$\alpha$ [packets/slot]},
yticklabel pos=right
]
\addlegendimage{legend columns=2, style={color=black, dotted, line width=2.0pt]}}
\addlegendentry{\footnotesize$\theta = -5 \text{ dB}$}
\addlegendimage{style={color=black, line width=2.0pt}}
\addlegendentry{\footnotesize$\theta = 0 \text{ dB}$}
\addlegendimage{style={color=black,
		mark=*, mark options={solid}, line width=2.0pt}}
\addlegendentry{\footnotesize$\theta = 5 \text{ dB}$}

\addplot [color=black, only marks, mark size=3pt, mark=*, mark options={solid, fill=green}]
table[row sep=crcr]{%
	0.15	2.13601130109271\\
};\addlegendentry{\footnotesize{Stability point}}

\addplot [color=mycolor2, mark=*, mark options={solid},line width=2.0pt]
table[row sep=crcr]{%
	0.05	1.17360520490003\\
	0.07	1.26722974181006\\
	0.09	1.38401587897303\\
	0.11	1.53748731610666\\
	0.13	1.75772224453968\\
	0.15	2.13601130109271\\
	0.17	inf\\
	0.19	inf\\
	0.21	inf\\
	0.23	inf\\
	0.25	inf\\
	0.27	inf\\
	0.29	inf\\
	0.31	inf\\
	0.33	inf\\
	0.35	inf\\
	0.37	inf\\
	0.39	inf\\
	0.41	inf\\
	0.43	inf\\
	0.45	inf\\
	0.47	inf\\
	0.49	inf\\
	0.51	inf\\
	0.53	inf\\
	0.55	inf\\
	0.57	inf\\
	0.59	inf\\
	0.61	inf\\
	0.63	inf\\
	0.65	inf\\
	0.67	inf\\
	0.69	inf\\
	0.71	inf\\
	0.73	inf\\
	0.75	inf\\
	0.77	inf\\
	0.79	inf\\
};
\addplot [color=black, only marks, mark size=3pt, mark=*, mark options={solid, fill=green}]
table[row sep=crcr]{%
	0.35	4.48216917628451\\
};
\addplot [color=black, only marks, mark size=3pt, mark=*, mark options={solid, fill=green}]
table[row sep=crcr]{%
	0.15	2.13601130109271\\
};

\addplot [color=mycolor2, line width=2.0pt]
table[row sep=crcr]{%
	0.05	1.09615375938418\\
	0.07	1.13935830606912\\
	0.09	1.18592187637567\\
	0.11	1.23641776944209\\
	0.13	1.29158608006148\\
	0.15	1.35235526369897\\
	0.17	1.42002093006921\\
	0.19	1.49631010569694\\
	0.21	1.58362167583047\\
	0.23	1.68563005643157\\
	0.25	1.80794772258687\\
	0.27	1.9599426257718\\
	0.29	2.15950128915614\\
	0.31	2.44583173683114\\
	0.33	2.93680803069121\\
	0.35	4.48216917628451\\
	0.37	inf\\
	0.39	inf\\
	0.41	inf\\
	0.43	inf\\
	0.45	inf\\
	0.47	inf\\
	0.49	inf\\
	0.51	inf\\
	0.53	inf\\
	0.55	inf\\
	0.57	inf\\
	0.59	inf\\
	0.61	inf\\
	0.63	inf\\
	0.65	inf\\
	0.67	inf\\
	0.69	inf\\
	0.71	inf\\
	0.73	inf\\
	0.75	inf\\
	0.77	inf\\
	0.79	inf\\
};

\addplot [color=mycolor2, dotted, line width=2.0pt]
table[row sep=crcr]{%
	0.05	1.06612918647197\\
	0.07	1.09372447480384\\
	0.09	1.12204582214216\\
	0.11	1.15115125126206\\
	0.13	1.18110523391927\\
	0.15	1.21198065279054\\
	0.17	1.24386170039379\\
	0.19	1.27684568542407\\
	0.21	1.311041837097\\
	0.23	1.34658104244248\\
	0.25	1.38360930045667\\
	0.27	1.42231570170467\\
	0.29	1.46291350035142\\
	0.31	1.50565963631341\\
	0.33	1.5508728051319\\
	0.35	1.59894494001488\\
	0.37	1.65036779796704\\
	0.39	1.70577751842421\\
	0.41	1.76600675985318\\
	0.43	1.83213814005848\\
	0.45	1.90568622713971\\
	0.47	1.98889854047503\\
	0.49	2.08497552375595\\
	0.51	2.1991781703361\\
	0.53	2.34049783957242\\
	0.55	2.5266598304001\\
	0.57	2.79998995512125\\
	0.59	3.30689162929792\\
	0.61	5.47158987551092\\
	0.63	inf\\
	0.65	inf\\
	0.67	inf\\
	0.69	inf\\
	0.71	inf\\
	0.73	inf\\
	0.75	inf\\
	0.77	inf\\
	0.79	inf\\
};

\addplot [color=black, only marks, mark size=3pt, mark=*, mark options={solid, fill=green}]
table[row sep=crcr]{%
	0.61	5.47158987551092\\
};

\end{axis}

\begin{axis}[%
every y tick label/.append style={font=\color{mycolor1}},
every y tick/.append style={mycolor1},
ymin=4,
ymax=22,
ylabel style={font=\color{mycolor1}},
ylabel={PAoI},
yticklabel pos=left,
xtick=\empty
]
\addplot [color=mycolor1, mark=*, mark options={solid}, line width=2.0pt]
table[row sep=crcr]{%
	0.05	21.1736052049\\
	0.07	15.5529440275243\\
	0.09	12.4951269900841\\
	0.11	10.6283964070157\\
	0.13	9.45002993684737\\
	0.15	8.80267796775938\\
	0.17	inf\\
	0.19	inf\\
	0.21	inf\\
	0.23	inf\\
	0.25	inf\\
	0.27	inf\\
	0.29	inf\\
	0.31	inf\\
	0.33	inf\\
	0.35	inf\\
	0.37	inf\\
	0.39	inf\\
	0.41	inf\\
	0.43	inf\\
	0.45	inf\\
	0.47	inf\\
	0.49	inf\\
	0.51	inf\\
	0.53	inf\\
	0.55	inf\\
	0.57	inf\\
	0.59	inf\\
	0.61	inf\\
	0.63	inf\\
	0.65	inf\\
	0.67	inf\\
	0.69	inf\\
	0.71	inf\\
	0.73	inf\\
	0.75	inf\\
	0.77	inf\\
	0.79	inf\\
};

\addplot [color=black, only marks, mark size=3pt, mark=*, mark options={solid, fill=green}]
table[row sep=crcr]{%
	0.15	8.80267796775938\\
};
\addplot [color=mycolor1, line width=2.0pt]
table[row sep=crcr]{%
	0.05	21.0961537593842\\
	0.07	15.4250725917834\\
	0.09	12.2970329874868\\
	0.11	10.3273268603512\\
	0.13	8.98389377236917\\
	0.15	8.01902193036564\\
	0.17	7.30237387124569\\
	0.19	6.75946800043378\\
	0.21	6.34552643773523\\
	0.23	6.03345614338809\\
	0.25	5.80794772258687\\
	0.27	5.66364632947551\\
	0.29	5.6077771512251\\
	0.31	5.67163818844404\\
	0.33	5.96711106099424\\
	0.35	7.33931203342737\\
	0.37	inf\\
	0.39	inf\\
	0.41	inf\\
	0.43	inf\\
	0.45	inf\\
	0.47	inf\\
	0.49	inf\\
	0.51	inf\\
	0.53	inf\\
	0.55	inf\\
	0.57	inf\\
	0.59	inf\\
	0.61	inf\\
	0.63	inf\\
	0.65	inf\\
	0.67	inf\\
	0.69	inf\\
	0.71	inf\\
	0.73	inf\\
	0.75	inf\\
	0.77	inf\\
	0.79	inf\\
};

\addplot [color=black, only marks, mark size=3pt, mark=*, mark options={solid, fill=green}]
table[row sep=crcr]{%
	0.35	7.33931203342737\\
};

\addplot [color=mycolor1, dotted, line width=2.0pt]
table[row sep=crcr]{%
	0.05	21.066129186472\\
	0.07	15.3794387605181\\
	0.09	12.2331569332533\\
	0.11	10.2420603421712\\
	0.13	8.87341292622696\\
	0.15	7.87864731945721\\
	0.17	7.12621464157026\\
	0.19	6.54000358016092\\
	0.21	6.07294659900176\\
	0.23	5.694407129399\\
	0.25	5.38360930045667\\
	0.27	5.12601940540837\\
	0.29	4.91118936242039\\
	0.31	4.73146608792632\\
	0.33	4.58117583543493\\
	0.35	4.45608779715774\\
	0.37	4.35307050066974\\
	0.39	4.26988008252678\\
	0.41	4.20503115009709\\
	0.43	4.15771953540732\\
	0.45	4.12790844936193\\
	0.47	4.11655811494311\\
	0.49	4.12579185028657\\
	0.51	4.15996248406159\\
	0.53	4.22729029240261\\
	0.55	4.34484164858192\\
	0.57	4.55437592003353\\
	0.59	5.00180688353521\\
	0.61	7.110934137806\\
	0.63	inf\\
	0.65	inf\\
	0.67	inf\\
	0.69	inf\\
	0.71	inf\\
	0.73	inf\\
	0.75	inf\\
	0.77	inf\\
	0.79	inf\\
};

\addplot [color=black, only marks, mark size=3pt, mark=*, mark options={solid, fill=green}]
table[row sep=crcr]{%
	0.61	7.110934137806\\
};
\end{axis}
\end{tikzpicture}%
	\fi
	\caption{PAoI (left) and average waiting time (right) for ET traffic with increasing arrival probability ($\alpha$) and  $\theta$.}
	\label{fig:AoI_WT_B} 
\end{figure}
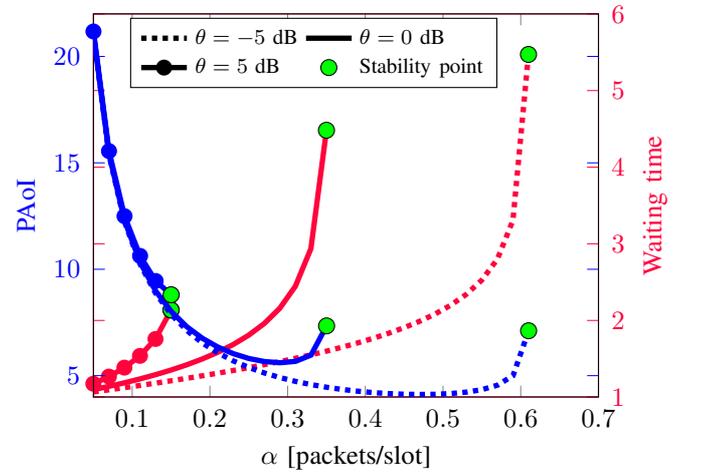

Similar to the \ac{TT} traffic, Fig. \ref{fig:AoI_WT_B} shows the \ac{PAoI} along with average waiting time for the \ac{ET} traffic with increasing arrival probability $\alpha$.
For low values of $\alpha$, the inter-arrival component dominates, yielding high values of \ac{PAoI}, while the waiting time is low. For low arrival probabilities, the network-wide aggregate interference is low, yielding higher probabilities for a packet to be successfully transmitted without large number of retransmissions. However, as $\alpha$ increases, the waiting times dominates, yielding an increase in the PAoI till point of queue instability, as indicated by the stability point.
\begin{figure} 
	\centering
	\ifCLASSOPTIONdraftcls
	\definecolor{mycolor2}{rgb}{0.24220,0.15040,0.66030}%
\definecolor{mycolor1}{rgb}{1,0.0,1.0}%
\definecolor{mycolor3}{rgb}{0.09640,0.75000,0.71204}%

\begin{tikzpicture}[scale=0.9]
\begin{groupplot}[group style={
	group name=myplot2,
	group size= 2 by 1, , horizontal sep=1.6cm}, height=2in,width=3in]

\nextgroupplot[title={{(a) TT traffic}},
scale only axis,
bar shift auto,
xmin=0.511111111111111,
xmax=10.4888888888889,
xtick={ 1,  2,  3,  4,  5,  6,  7,  8,  9, 10},
xlabel style={font=\color{white!15!black}},
xlabel={QoS class},
ymin=0,
ymax=24,
ylabel style={font=\color{white!15!black}},
ylabel={PAoI},
axis background/.style={fill=white},
legend style={legend columns=-1, legend cell align=left, align=left, draw=white!15!black}
]
\addplot[ybar, bar width=0.178, fill=mycolor1, draw=black, area legend] table[row sep=crcr] {%
	1	18.6765574223332\\
	2	18.43898864065\\
	3	18.2791523549356\\
	4	18.2710223820915\\
	5	18.2681357706962\\
	6	18.2671048015103\\
	7	18.2667699626074\\
	8	18.2666824240599\\
	9	18.2666676137498\\
	10	18.2666666690103\\
};\addlegendentry{$T = 15$}

\addplot[ybar, bar width=0.178, fill=red, draw=black, area legend] table[row sep=crcr] {%
	1	15.227690185811\\
	2	14.2003065971877\\
	3	13.918055595058\\
	4	13.6532020948606\\
	5	13.6550633019502\\
	6	13.6633685792002\\
	7	13.4034234088582\\
	8	13.4011104381432\\
	9	13.4002192744829\\
	10	13.4000076071209\\
};\addlegendentry{$T=10$}

\addplot[ybar, bar width=0.178, fill=blue, draw=black, area legend] table[row sep=crcr] {%
	1	inf\\
	2	inf\\
	3	inf\\
	4	inf\\
	5	inf\\
	6	inf\\
	7	inf\\
	8	inf\\
	9	19.9897897073938\\
	10	19.30271442752326\\
};\addlegendentry{$T=5$}

\nextgroupplot[title={{(b) ET traffic}},
scale only axis,
bar shift auto,
xmin=0.511111111111111,
xmax=10.4888888888889,
xtick={ 1,  2,  3,  4,  5,  6,  7,  8,  9, 10},
xlabel style={font=\color{white!15!black}},
xlabel={QoS class},
ymin=0,
ymax=20,
ylabel style={font=\color{white!15!black}},
ylabel={PAoI},
axis background/.style={fill=white},
legend style={legend columns=-1, legend cell align=left, align=left, draw=white!15!black}
]
\addplot[ybar, bar width=0.178, fill=mycolor1, draw=black, area legend] table[row sep=crcr] {%
	1	16.4058844959822\\
	2	16.3334436890329\\
	3	16.2965773596966\\
	4	16.2700993890142\\
	5	16.2483995128242\\
	6	16.2291417057973\\
	7	16.2109292752769\\
	8	16.192529358105\\
	9	16.1720784207814\\
	10	16.143497957223\\
};\addlegendentry{$\alpha =  1/15$}

\addplot[ybar, bar width=0.178, fill=red, draw=black, area legend] table[row sep=crcr] {%
	1	11.7898144315164\\
	2	11.526925992605\\
	3	11.3473243039206\\
	4	11.1914965487241\\
	5	11.0465752126802\\
	6	11.0073265450541\\
	7	10.99707440296838\\
	8	10.90343065629983\\
	9	11.8544326759826\\
	10	11.8298070293561\\
};\addlegendentry{$\alpha =  1/10$}

\addplot[ybar, bar width=0.178, fill=blue, draw=black, area legend] table[row sep=crcr] {%
	1	inf\\
	2	inf\\
	3	inf\\
	4	inf\\
	5	inf\\
	6	inf\\
	7	inf\\
	8	15.3505321737371\\
	9	10.7968604626857\\
	10	8.29961092528308\\
};\addlegendentry{$\alpha =  1/5$}

\end{groupplot}
\end{tikzpicture}	
	\vspace{-1cm}
	\else
	\definecolor{mycolor2}{rgb}{0.24220,0.15040,0.66030}%
\definecolor{mycolor1}{rgb}{1,0.0,1.0}%
\definecolor{mycolor3}{rgb}{0.09640,0.75000,0.71204}%

\begin{tikzpicture}[scale=0.9]
\begin{groupplot}[group style={
	group name=myplot2,
	group size= 1 by 2, , vertical sep=1.8cm}, height=2in,width=0.8\columnwidth,]

\nextgroupplot[title={{(a) TT traffic}},
scale only axis,
bar shift auto,
xmin=0.511111111111111,
xmax=10.4888888888889,
xtick={ 1,  2,  3,  4,  5,  6,  7,  8,  9, 10},
xlabel style={font=\color{white!15!black}},
xlabel={QoS class},
ymin=0,
ymax=23.5,
ylabel style={font=\color{white!15!black}},
ylabel={PAoI},
axis background/.style={fill=white},
legend style={legend columns=-1, legend cell align=left, align=left, draw=white!15!black}
]
\addplot[ybar, bar width=0.178, fill=mycolor1, draw=black, area legend] table[row sep=crcr] {%
	1	18.6765574223332\\
	2	18.43898864065\\
	3	18.2791523549356\\
	4	18.2710223820915\\
	5	18.2681357706962\\
	6	18.2671048015103\\
	7	18.2667699626074\\
	8	18.2666824240599\\
	9	18.2666676137498\\
	10	18.2666666690103\\
};\addlegendentry{$T = 15$}

\addplot[ybar, bar width=0.178, fill=red, draw=black, area legend] table[row sep=crcr] {%
	1	15.227690185811\\
	2	14.2003065971877\\
	3	13.918055595058\\
	4	13.6532020948606\\
	5	13.6550633019502\\
	6	13.6633685792002\\
	7	13.4034234088582\\
	8	13.4011104381432\\
	9	13.4002192744829\\
	10	13.4000076071209\\
};\addlegendentry{$T=10$}

\addplot[ybar, bar width=0.178, fill=blue, draw=black, area legend] table[row sep=crcr] {%
	1	inf\\
	2	inf\\
	3	inf\\
	4	inf\\
	5	inf\\
	6	inf\\
	7	inf\\
	8	inf\\
	9	19.9897897073938\\
	10	19.30271442752326\\
};\addlegendentry{$T=5$}

\nextgroupplot[title={{(b) ET traffic}},
scale only axis,
bar shift auto,
xmin=0.511111111111111,
xmax=10.4888888888889,
xtick={ 1,  2,  3,  4,  5,  6,  7,  8,  9, 10},
xlabel style={font=\color{white!15!black}},
xlabel={QoS class},
ymin=0,
ymax=20,
ylabel style={font=\color{white!15!black}},
ylabel={PAoI},
axis background/.style={fill=white},
legend style={legend columns=-1, legend cell align=left, align=left, draw=white!15!black}
]
\addplot[ybar, bar width=0.178, fill=mycolor1, draw=black, area legend] table[row sep=crcr] {%
	1	16.4058844959822\\
	2	16.3334436890329\\
	3	16.2965773596966\\
	4	16.2700993890142\\
	5	16.2483995128242\\
	6	16.2291417057973\\
	7	16.2109292752769\\
	8	16.192529358105\\
	9	16.1720784207814\\
	10	16.143497957223\\
};\addlegendentry{$\alpha =  1/15$}

\addplot[ybar, bar width=0.178, fill=red, draw=black, area legend] table[row sep=crcr] {%
	1	11.7898144315164\\
	2	11.526925992605\\
	3	11.3473243039206\\
	4	11.1914965487241\\
	5	11.0465752126802\\
	6	11.0073265450541\\
	7	10.99707440296838\\
	8	10.90343065629983\\
	9	11.8544326759826\\
	10	11.8298070293561\\
};\addlegendentry{$\alpha =  1/10$}

\addplot[ybar, bar width=0.178, fill=blue, draw=black, area legend] table[row sep=crcr] {%
	1	inf\\
	2	inf\\
	3	inf\\
	4	inf\\
	5	inf\\
	6	inf\\
	7	inf\\
	8	15.3505321737371\\
	9	10.7968604626857\\
	10	8.29961092528308\\
};\addlegendentry{$\alpha =  1/5$}

\end{groupplot}
\end{tikzpicture}	
	\fi
	\caption{\ac{PAoI} for $N = 10$ QoS classes and $\theta = 5$ dB.}
	\label{fig:AoI_perClass} 
\end{figure}
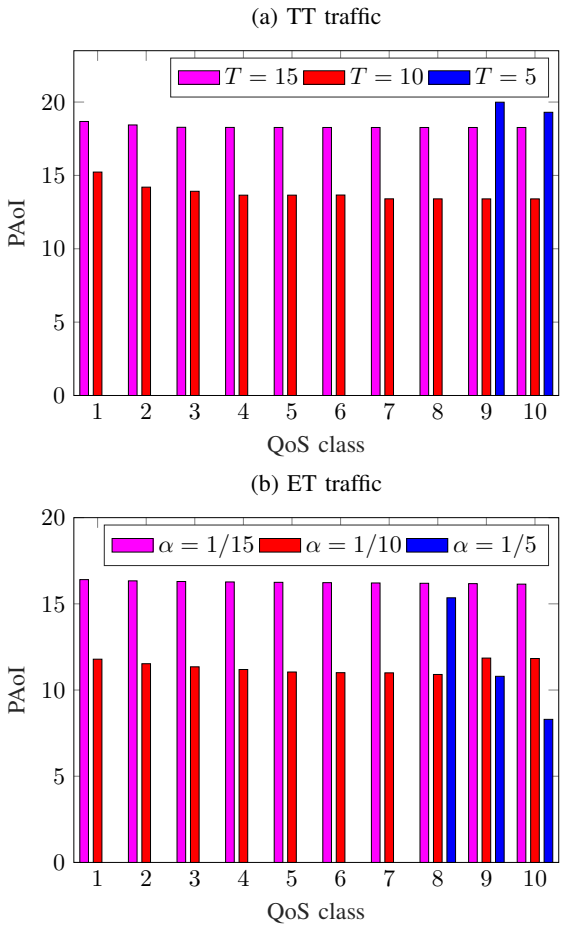 

Fig. \ref{fig:AoI_perClass} presents the per-QoS class \ac{PAoI} among the different \ac{QoS} classes within the network. The shown classes are sorted in an ascending order with respect to $d_n$ (i.e., a device belonging to class $i$ is spatially located closer to its serving \ac{BS} compared to a device belonging to class $j$, such that $i>j$). The \ac{TT} and \ac{ET} traffic models are shown in Fig. \ref{fig:AoI_perClass}(a) and Fig. \ref{fig:AoI_perClass}(b), respectively. For $T =15 \;(\alpha=1/15)$, the inter-arrival times dominates the \ac{PAoI}, leading to a nearly-constant \ac{PAoI} over all the classes. The location-dependency is more clear as $T (\alpha)$ decreases (increases).  for $\alpha=0.15$ and $\alpha=0.25$. Consequently, classes with lower indices experience large \ac{PAoI} due to their larger waiting times (i.e., effect of the location dependency captured via the meta distribution). For large traffic load (i.e., $ T = 5 \; (\alpha = 1/5$)), all except last two and three classes are unstable, for the \ac{TT} and \ac{ET} traffic models, respectively. As mentioned earlier, unstable queues results in infinite \ac{PAoI}.

\begin{figure*}
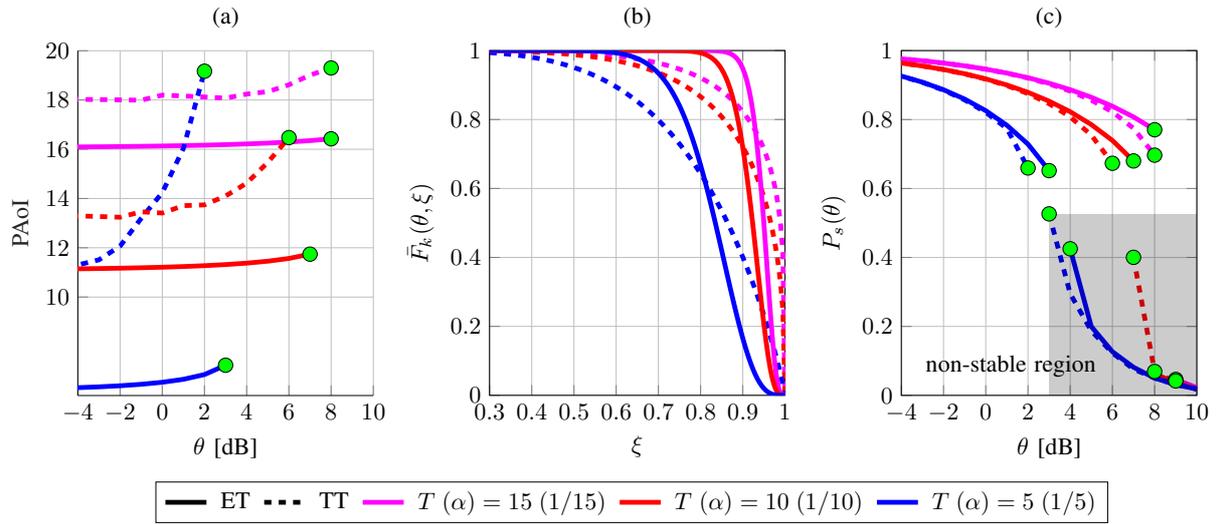

	\centering
	\ifCLASSOPTIONdraftcls
	\input{simulation_results/figures/det_geo_comp/1col/det_geo_comp.tex}	
	\vspace{-0.1in}
	\else
	\input{simulation_results/figures/det_geo_comp/1col/det_geo_comp.tex}	
	\fi
	\caption{TT and ET traffic models comparison based on (a) PAoI (b) meta distribution for $\theta = 1$ (c) TSP.}
	\label{fig:det_geo_comp} 
\end{figure*}
\begin{figure} 
	\centering
	\ifCLASSOPTIONdraftcls
	\definecolor{mycolor1}{rgb}{0.00000,1.00000,1.00000}%
\definecolor{mycolor2}{rgb}{0.00000,0.44700,0.74100}%
\definecolor{mycolor3}{rgb}{0.85000,0.32500,0.09800}%
\definecolor{mycolor4}{rgb}{0.49020,0.18039,0.56078}%
\definecolor{mycolor5}{rgb}{0.55,0.71,0.0}

\begin{tikzpicture}
\begin{groupplot}[group style={
	group name=myplot2,
	group size= 2 by 1, , horizontal sep=2cm}, height=2in,width=2.5in]
\nextgroupplot[title={{\smash{(a) TT traffic}}},
scale only axis,
xmin=-2,
xmax=10,
xlabel={$\theta$ [dB]},
ymin=4,
ymax=18,
y dir=reverse,
ylabel={$T$},
axis x line*=bottom,
axis y line*=left,
ytick={18,16,14,12,10,8,6,4}, 
grid=both]
\addplot[draw=none,name path=B] {18};     

\addplot[color=black, name path=data1] table[row sep=crcr]{%
	-5	4\\
	-4	4\\
	-3	4\\
	-2	4\\
	-1	4\\
	0	4\\
	1	4\\
	2	4\\
	3	4\\
	4	4.2\\
	5	4.8\\
	6	5.6\\
	7	7\\
	8	9.4\\
	9	11.6666666666667\\
	10	16\\
	12   18\\
};
\addplot[fill=mycolor1] fill between[of=data1 and B, soft clip={domain=-4:12}]; 

\addplot[color=black, name path=data2] table[row sep=crcr]{%
	-5	4\\
	-4	4\\
	-3	4\\
	-2	4\\
	-1	4\\
	0	4\\
	1	4\\
	2	4.2\\
	3	4.6\\
	4	5.2\\
	5	6.2\\
	6	7.6\\
	7	9.6\\
	8	12\\
	9	14.6666666666667\\
	10	18\\
};

\addplot[fill=mycolor2] fill between[of=data2 and B, soft clip={domain=-4:12}]; 

\addplot[color=black, name path=data3] table[row sep=crcr]{%
	-5	4\\
	-4	4\\
	-3	4\\
	-2	4\\
	-1	4\\
	0	4\\
	1	4.2\\
	2	4.4\\
	3	4.8\\
	4	5.6\\
	5	6.8\\
	6	8.2\\
	7	10.2\\
	8	12.6\\
	9	15\\
	10	18\\
};
\addplot[fill=mycolor3] fill between[of=data3 and B, soft clip={domain=-4:12}]; 

\addplot[color=black, name path=data4] table[row sep=crcr]{%
	-5	4\\
	-4	4\\
	-3	4\\
	-2	4\\
	-1	4\\
	0	4.2\\
	1	4.4\\
	2	4.8\\
	3	5.4\\
	4	6.2\\
	5	7.2\\
	6	8.6\\
	7	10.4\\
	8	12.6\\
	9	15\\
	10	18\\
};
\addplot[fill=mycolor4] fill between[of=data4 and B, soft clip={domain=-4:12}]; 

\addplot[color=black, dotted, line width=2.0pt, name path=data5] table[row sep=crcr]{%
	-5	4\\
	-4	4\\
	-3	4\\
	-2	4\\
	-1	4\\
	0	4.2\\
	1	4.6\\
	2	5\\
	3	5.6\\
	4	6.4\\
	5	7.4\\
	6	8.6\\
	7	10.4\\
	8	12.6\\
	9	15\\
	10	18\\
};
\addplot[fill=green,fill opacity=0.5] fill between[of=data5 and B]; 

\draw[->, line width=0.5mm](rel axis cs:0.5,0.5) -- (rel axis cs:0.8,0.9);
\node [] at (rel axis cs:0.48,0.44) {Increasing QoS class};

\nextgroupplot[title={{(b) ET traffic}},
scale only axis,
xmin=-5,
xmax=10,
xlabel={$\theta$ [dB]},
ymin=0.05,
ymax=0.8,
ylabel={$\alpha$},
ytick={0.1,0.2,0.3,0.4,0.5,0.6,0.7,0.8}, 
ylabel style={font=\color{white!15!black}},
xlabel style={font=\color{white!15!black}},
axis x line*=bottom,
axis y line*=left,
xmajorgrids,
ymajorgrids,
grid=both]

\addplot[fill=red,fill opacity=0.3, draw=black] table[row sep=crcr]{%
	-5	0.91\\
	-4	0.876666666666667\\
	-3	0.842\\
	-2	0.802\\
	-1	0.758\\
	0	0.702\\
	1	0.642\\
	2	0.574\\
	3	0.502\\
	4	0.426\\
	5	0.354\\
	6	0.282\\
	7	0.218\\
	8	0.162\\
	9	0.11\\
	10	0.07\\
}
\closedcycle;

\addplot[fill=red,fill opacity=0.4, draw=black] table[row sep=crcr]{%
	-5	0.83\\
	-4	0.79\\
	-3	0.746\\
	-2	0.694\\
	-1	0.638\\
	0	0.578\\
	1	0.51\\
	2	0.442\\
	3	0.374\\
	4	0.306\\
	5	0.242\\
	6	0.186\\
	7	0.138\\
	8	0.102\\
	9	0.07\\
	10	0.05\\
}
\closedcycle;

\addplot[fill=red,fill opacity=0.5, draw=black] table[row sep=crcr]{%
	-5	0.79\\
	-4	0.736666666666667\\
	-3	0.682\\
	-2	0.626\\
	-1	0.566\\
	0	0.502\\
	1	0.438\\
	2	0.374\\
	3	0.31\\
	4	0.254\\
	5	0.202\\
	6	0.158\\
	7	0.122\\
	8	0.094\\
	9	0.07\\
	10	0.05\\
}
\closedcycle;

\addplot[fill=red,fill opacity=0.6, draw=black] table[row sep=crcr]{%
	-5	0.73\\
	-4	0.676666666666667\\
	-3	0.622\\
	-2	0.566\\
	-1	0.51\\
	0	0.45\\
	1	0.39\\
	2	0.33\\
	3	0.274\\
	4	0.226\\
	5	0.182\\
	6	0.146\\
	7	0.118\\
	8	0.094\\
	9	0.07\\
	10	0.05\\
}
\closedcycle;

\addplot[fill=red,fill opacity=0.7, draw=black] table[row sep=crcr]{%
	-5	0.65\\
	-4	0.603333333333333\\
	-3	0.546\\
	-2	0.494\\
	-1	0.438\\
	0	0.386\\
	1	0.338\\
	2	0.294\\
	3	0.25\\
	4	0.214\\
	5	0.178\\
	6	0.146\\
	7	0.118\\
	8	0.094\\
	9	0.07\\
	10	0.05\\
}
\closedcycle;

\addplot [color=black, dotted, line width=2.0pt]table[row sep=crcr]{%
	-5	0.65\\
	-4	0.603333333333333\\
	-3	0.546\\
	-2	0.494\\
	-1	0.438\\
	0	0.386\\
	1	0.338\\
	2	0.294\\
	3	0.25\\
	4	0.214\\
	5	0.178\\
	6	0.146\\
	7	0.118\\
	8	0.094\\
	9	0.07\\
	10	0.05\\
};

\draw[->, line width=0.5mm](rel axis cs:0.4,0.23) -- (rel axis cs:0.7,0.5);
\node [] at (rel axis cs:0.35,0.17) {Increasing QoS class};

\end{groupplot}

\end{tikzpicture}		
	\else
	\definecolor{mycolor1}{rgb}{0.00000,1.00000,1.00000}%
\definecolor{mycolor2}{rgb}{0.00000,0.44700,0.74100}%
\definecolor{mycolor3}{rgb}{0.85000,0.32500,0.09800}%
\definecolor{mycolor4}{rgb}{0.49020,0.18039,0.56078}%
\definecolor{mycolor5}{rgb}{0.55,0.71,0.0}

\begin{tikzpicture}
\begin{groupplot}[group style={
	group name=myplot2,
	group size= 1 by 2, , vertical sep=2.4cm}, height=2in,width=0.7\columnwidth]
\nextgroupplot[title={{\smash{(a) TT traffic}}},
scale only axis,
xmin=-2,
xmax=10,
xlabel={$\theta$ [dB]},
ymin=4,
ymax=18,
y dir=reverse,
ylabel={$T$},
axis x line*=bottom,
axis y line*=left,
ytick={18,16,14,12,10,8,6,4}, 
grid=both]
\addplot[draw=none,name path=B] {18};     

\addplot[color=black, name path=data1] table[row sep=crcr]{%
	-5	4\\
	-4	4\\
	-3	4\\
	-2	4\\
	-1	4\\
	0	4\\
	1	4\\
	2	4\\
	3	4\\
	4	4.2\\
	5	4.8\\
	6	5.6\\
	7	7\\
	8	9.4\\
	9	11.6666666666667\\
	10	16\\
	12   18\\
};
\addplot[fill=mycolor1] fill between[of=data1 and B, soft clip={domain=-4:12}]; 

\addplot[color=black, name path=data2] table[row sep=crcr]{%
	-5	4\\
	-4	4\\
	-3	4\\
	-2	4\\
	-1	4\\
	0	4\\
	1	4\\
	2	4.2\\
	3	4.6\\
	4	5.2\\
	5	6.2\\
	6	7.6\\
	7	9.6\\
	8	12\\
	9	14.6666666666667\\
	10	18\\
};

\addplot[fill=mycolor2] fill between[of=data2 and B, soft clip={domain=-4:12}]; 

\addplot[color=black, name path=data3] table[row sep=crcr]{%
	-5	4\\
	-4	4\\
	-3	4\\
	-2	4\\
	-1	4\\
	0	4\\
	1	4.2\\
	2	4.4\\
	3	4.8\\
	4	5.6\\
	5	6.8\\
	6	8.2\\
	7	10.2\\
	8	12.6\\
	9	15\\
	10	18\\
};
\addplot[fill=mycolor3] fill between[of=data3 and B, soft clip={domain=-4:12}]; 

\addplot[color=black, name path=data4] table[row sep=crcr]{%
	-5	4\\
	-4	4\\
	-3	4\\
	-2	4\\
	-1	4\\
	0	4.2\\
	1	4.4\\
	2	4.8\\
	3	5.4\\
	4	6.2\\
	5	7.2\\
	6	8.6\\
	7	10.4\\
	8	12.6\\
	9	15\\
	10	18\\
};
\addplot[fill=mycolor4] fill between[of=data4 and B, soft clip={domain=-4:12}]; 

\addplot[color=black, dotted, line width=2.0pt, name path=data5] table[row sep=crcr]{%
	-5	4\\
	-4	4\\
	-3	4\\
	-2	4\\
	-1	4\\
	0	4.2\\
	1	4.6\\
	2	5\\
	3	5.6\\
	4	6.4\\
	5	7.4\\
	6	8.6\\
	7	10.4\\
	8	12.6\\
	9	15\\
	10	18\\
};
\addplot[fill=green,fill opacity=0.5] fill between[of=data5 and B]; 

\draw[->, line width=0.5mm](rel axis cs:0.5,0.5) -- (rel axis cs:0.8,0.9);
\node [] at (rel axis cs:0.48,0.44) {Increasing QoS class};

\nextgroupplot[title={{(b) ET traffic}},
scale only axis,
xmin=-5,
xmax=10,
xlabel={$\theta$ [dB]},
ymin=0.05,
ymax=0.8,
ylabel={$\alpha$},
ytick={0.1,0.2,0.3,0.4,0.5,0.6,0.7,0.8}, 
ylabel style={font=\color{white!15!black}},
xlabel style={font=\color{white!15!black}},
axis x line*=bottom,
axis y line*=left,
xmajorgrids,
ymajorgrids,
grid=both]

\addplot[fill=red,fill opacity=0.3, draw=black] table[row sep=crcr]{%
	-5	0.91\\
	-4	0.876666666666667\\
	-3	0.842\\
	-2	0.802\\
	-1	0.758\\
	0	0.702\\
	1	0.642\\
	2	0.574\\
	3	0.502\\
	4	0.426\\
	5	0.354\\
	6	0.282\\
	7	0.218\\
	8	0.162\\
	9	0.11\\
	10	0.07\\
}
\closedcycle;

\addplot[fill=red,fill opacity=0.4, draw=black] table[row sep=crcr]{%
	-5	0.83\\
	-4	0.79\\
	-3	0.746\\
	-2	0.694\\
	-1	0.638\\
	0	0.578\\
	1	0.51\\
	2	0.442\\
	3	0.374\\
	4	0.306\\
	5	0.242\\
	6	0.186\\
	7	0.138\\
	8	0.102\\
	9	0.07\\
	10	0.05\\
}
\closedcycle;

\addplot[fill=red,fill opacity=0.5, draw=black] table[row sep=crcr]{%
	-5	0.79\\
	-4	0.736666666666667\\
	-3	0.682\\
	-2	0.626\\
	-1	0.566\\
	0	0.502\\
	1	0.438\\
	2	0.374\\
	3	0.31\\
	4	0.254\\
	5	0.202\\
	6	0.158\\
	7	0.122\\
	8	0.094\\
	9	0.07\\
	10	0.05\\
}
\closedcycle;

\addplot[fill=red,fill opacity=0.6, draw=black] table[row sep=crcr]{%
	-5	0.73\\
	-4	0.676666666666667\\
	-3	0.622\\
	-2	0.566\\
	-1	0.51\\
	0	0.45\\
	1	0.39\\
	2	0.33\\
	3	0.274\\
	4	0.226\\
	5	0.182\\
	6	0.146\\
	7	0.118\\
	8	0.094\\
	9	0.07\\
	10	0.05\\
}
\closedcycle;

\addplot[fill=red,fill opacity=0.7, draw=black] table[row sep=crcr]{%
	-5	0.65\\
	-4	0.603333333333333\\
	-3	0.546\\
	-2	0.494\\
	-1	0.438\\
	0	0.386\\
	1	0.338\\
	2	0.294\\
	3	0.25\\
	4	0.214\\
	5	0.178\\
	6	0.146\\
	7	0.118\\
	8	0.094\\
	9	0.07\\
	10	0.05\\
}
\closedcycle;

\addplot [color=black, dotted, line width=2.0pt]table[row sep=crcr]{%
	-5	0.65\\
	-4	0.603333333333333\\
	-3	0.546\\
	-2	0.494\\
	-1	0.438\\
	0	0.386\\
	1	0.338\\
	2	0.294\\
	3	0.25\\
	4	0.214\\
	5	0.178\\
	6	0.146\\
	7	0.118\\
	8	0.094\\
	9	0.07\\
	10	0.05\\
};

\draw[->, line width=0.5mm](rel axis cs:0.4,0.23) -- (rel axis cs:0.7,0.5);
\node [] at (rel axis cs:0.37,0.17) {Increasing QoS class};

\end{groupplot}

\end{tikzpicture}		
	\fi
	\caption{Pareto frontiers between detection threshold and traffic load for $N=5$. Dashed lines represent the spatially averaged frontiers.}
	\label{fig:Pareto} 
\end{figure}
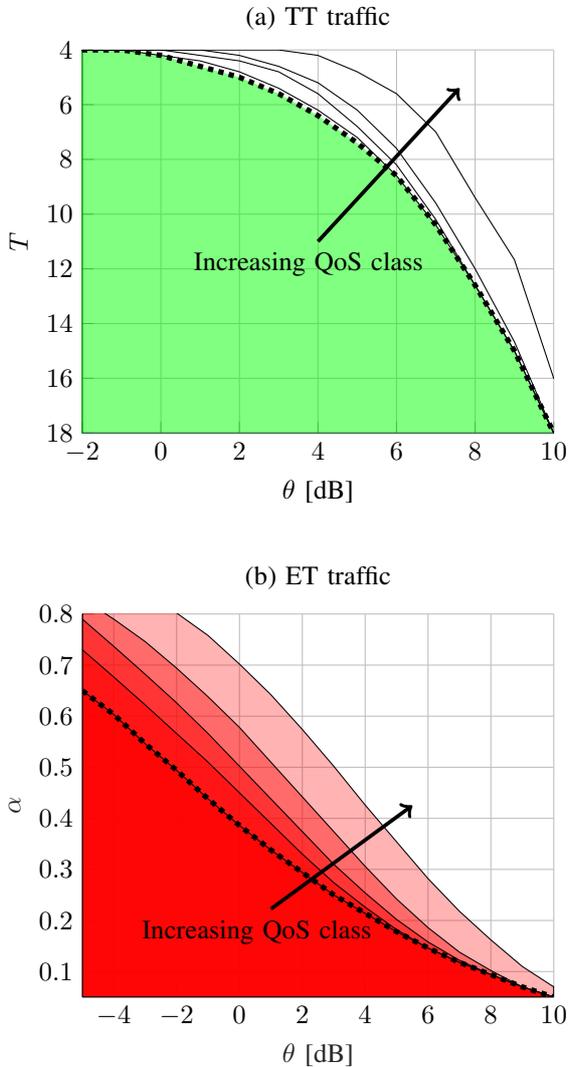 
Next we assess the \ac{TT} and \ac{ET} traffic based on their \ac{PAoI}, meta distribution, and \acp{TSP} for three different traffic loads. Fig. \ref{fig:det_geo_comp} presents different performance comparisons between the two traffic models. First, Fig. \ref{fig:det_geo_comp}(a) shows the \ac{PAoI} as a function of increasing detection threshold $\theta$. It is observed that the \ac{ET} traffic provides lower \ac{PAoI} for all the considered set of traffic loads. Although, it was shown in \cite{Talak2018} that periodic packet generation  minimizes the age for the \ac{FCFS} queues, considering the network-wide aggregate interference into the age analysis provides another perspective. As mentioned in Section \ref{sec:sg_analysis}, the \ac{TT} traffic model imposes a spatial and temporal correlation between the devices. In particular, each device sees the same set of active (i.e., interfering) devices in each transmission cycle $T$. Such correlation is alleviated in the \ac{ET} traffic model, in which the activity profiles are diversified among different time slots. This performance gap between the \ac{TT} and \ac{ET} traffic is larger for low duty cycles (or high arrival probabilities), due to the stronger interference correlation in such scenarios. As the activity profiles are more relaxed (i.e., $T$ $(\alpha)$ increases (decreases)), the gap between the two traffic models decreases. Next, Fig. \ref{fig:det_geo_comp}(b) presents the meta distribution for the considered traffic loads. As the traffic load increases for the two traffic models, the percentile of devices achieving a given reliability (i.e., $\xi$) decreases as explained in Fig. \ref{fig:meta_sim}. In addition, one can observe the discrepancies  between the \ac{TT} and \ac{ET} traffic considering the similar traffic load. Such discrepancies are hardly captured by the spatially averages $P_s(\theta)$, which emphasizes the importance of the meta distribution as shown in \ref{fig:det_geo_comp}(c). In addition, a sharper transition in the meta distribution implies less location-dependent performance (i.e., less temporal interference correlation) and that all devices tend to operate as a typical device. Due to the aforementioned explained correlation between the active devices, the \ac{TT} traffic provides lower \acp{TSP} for  all the considered traffic loads. The stability point, depicted by green circles, represent the point at which the queues are unstable. Any operation beyond such a point yields in operating in the non-stable region. 

Finally, Fig. \ref{fig:Pareto} presents the Pareto frontiers for the arrival intensity of the \ac{ET} and \ac{TT} traffic with the detection threshold over the $N$ QoS classes. Pareto frontiers define regions where the queues are guaranteed to be operating within a stable region. First, Fig. \ref{fig:Pareto}(a) shows the relation between the arrival probability and the detection threshold $\theta$ for the existing five QoS classes. Due to retransmissions, a higher $\theta$ implies lower idle probability, and hence, higher aggregate network interference allowing lower values of $\alpha$ to ensure stability. In addition, due to the favorable spatial locations of the higher QoS classes compared to the lower ones, the Pareto frontiers for those higher classes are covering a larger set of $(\theta, \alpha)$ values. Similarly, Fig. \ref{fig:Pareto}(b) presents the Pareto frontiers between $\theta$ and the cyclic time $T$. The curves explanation follows that of the \ac{ET} traffic, since $T$ represents the arrival events, comparable to $\alpha$.

	\thispagestyle{empty}
\section{Conclusion}\label{sec:Conclusion}

This paper presents a mathematical spatiotemporal framework to characterize the peak age of information (\ac{PAoI}) in \ac{IoT} uplink networks for time-triggered (TT) and event-triggered (ET) traffic. First, we leverage tools from stochastic geometry to analyze the location-dependent performance of the network under the two traffic models. Expressions for the network-wide aggregate interference are presented in the context of the location-aware meta distribution. Furthermore, we analyze the interdependency between the aggregate network wide-interference and the queues evolution at each device.  Additionally, a novel PH/Geo/1 queueing model is proposed to model the periodic traffic generation for the \ac{TT} traffic at each device. Expressions for the average waiting time and \ac{PAoI} are derived to the TT and ET traffic models. To this end, simulation results are presented to validate the proposed framework. The results unveil the counterintuitive lower PAoI of ET traffic over the TT traffic, which is due to the higher temporal interference correlations of the TT traffic. In addition, the stability frontiers coupling the network's traffic load and decoding threshold are presented and their effect on the PAoI are discussed.
	\thispagestyle{empty}
\bibliographystyle{./lib/IEEEtran.cls}
\bibliography{./literature/Literature_Local}

	\thispagestyle{empty}
	\appendices

\section{Proof of Lemma 1}\label{se:Appendix_A}

The $b$-th moment of the \ac{TSP} can be derived from eq.(\ref{eq:PS}) as
\begin{equation}\label{app1}
M_b = \mathbb{E}^{!}_{r_i,P_i,r_o} \Big[ \prod_{\omega_i\in \tilde{\mathrm{\Phi}}_{T}}  \Big(\frac{1}{1 + \frac{\theta r_o^{\eta(1-\epsilon)}}{\rho \omega_i}} \Big)^b\Big| \hat{\mathrm{\Phi}}, \mathrm{\Psi} \Big], 
\end{equation}
where, the uplink transmission power $P_i$ of the $i$-th device is a random variable due to the employed fractional path-loss power control \cite{ElSawy2017}. In (\ref{app1}), the average is first conditioned on $r_o$ then evaluated via the probability generating functional of the \ac{PPP} with the intensity function $\tilde{\lambda}_T(\omega)$. The distribution of $r_o$ is given by $f_{r_o}(r) = 2\pi\lambda r e^{-\pi \lambda r^2}$. With some mathematical operations following \cite{ElSawy_meta}, the lemma is proved.  

\section{Proof of Lemma 4}\label{se:Appendix_B}
Based on \cite{G.Kulkarni1999},\cite{Alfa2015}, $\mathbf{R}_n$ is the minimal non-negative solution to the quadratic equation $\mathbf{R}_n = \mathbf{A}_{0,n} + \mathbf{R}_n\mathbf{A}_{1,n} + \mathbf{R}_n^2\mathbf{A}_{2,n}$. Let $\mathbf{x}_{0,n}$ and $\mathbf{x}_{1,n}$ be the solution to 
\begin{equation}\label{eq:xox1_MAM}
\Big[\mathbf{x}_{0,n} \;\; \mathbf{x}_{1,n}\Big] = \Big[\mathbf{x}_{i,n} \;\; \mathbf{x}_{i,n}\Big] \left[ \begin{array}{ll} \mathbf{B}_{1,n} & \mathbf{C} \\ \mathbf{A}_{2,n} & \mathbf{A}_{1,n}+\mathbf{R}_n\mathbf{A}_{2,n} \end{array}\right].
\end{equation}
Since $\mathbf{A}_{0,n}$ is rank 1, $\mathbf{R}_n$ can be rewritten as 
\begin{equation}
\mathbf{R}_n= \mathbf{A}_{0,n}(\mathbf{I}_{T}-\mathbf{A}_{1,n}-\mathbf{A}_{2,n}\mathbf{G}_2)^{-1},
\end{equation}
where $\mathbf{G}_n$ is the minimal non-negative solution to $\mathbf{G}_n = \mathbf{A}_{2,n} + \mathbf{A}_{1,n}\mathbf{G}_n + \mathbf{A}_{0,n}\mathbf{G}_n^2$. Following \cite{Alfa2015}[Chapter 5.9], the lemma can be proved.
	
\end{document}